\documentclass[journal,twoside,web]{ieeecolor}
\usepackage{generic}
\usepackage{cite}
\usepackage{amsmath,amssymb,amsfonts}
\usepackage{algorithmic}
\usepackage{graphicx}
\usepackage{algorithm,algorithmic}
\usepackage{hyperref}
\hypersetup{hidelinks=true}
\usepackage{textcomp}

\usepackage{breqn}
\usepackage{graphics}
\usepackage{mathtools, cuted}
\usepackage{lipsum}
\usepackage{stfloats}
\usepackage{multirow}
\usepackage{wrapfig}
\usepackage{euscript,psfrag,latexsym}
\usepackage{bbm,amstext,wasysym}

\usepackage[normalem]{ulem}
\graphicspath{{./},{./figures/}}
\usepackage{dsfont}
\usepackage{ulem}
\newcommand{\cL}{{\cal L}}

\newcommand{\cC}{{\cal C}}

\newcommand{\mR}{{\mathbb R}}

\usepackage[algo2e,linesnumbered,ruled]{algorithm2e}

\newcommand{\bD}{{\mathbf D}}

\newcommand{\bU}{{\mathbf U}}

\newcommand{\bs}{{\mathbf s}}
\newcommand{\bv}{{\mathbf v}}

	\usepackage[utf8]{inputenc}

\newcommand{\bF}{{\mathbf F}}
\newcommand{\bk}{{\mathbf k}}

\newcommand{\by}{{\mathbf y}}
\newcommand{\bff}{{\mathbf f}}

\newcommand{\bx}{{\mathbf x}}
\newcommand{\brho}{{\boldsymbol \rho}}

\newcommand{\bPsi}{{\boldsymbol \Psi}}
\newcommand{\bC}{{\boldsymbol{\mathcal C}}}

\newcommand{\bX}{{\mathbf X}}

\newcommand{\bg}{{\mathbf g}}
\newcommand{\bA}{{\mathbf A}}
\newcommand{\bB}{{\mathbf B}}
\newcommand{\bu}{{\mathbf u}}

\newcommand{\bz}{{\mathbf z}}

\DeclareMathOperator*{\argmin}{arg\,min}

\newcommand{\bJ}{{\mathbf J}}
\newcommand{\bH}{{\mathbf H}}

\usepackage{mathtools,halloweenmath}

\newtheorem{theorem}{Theorem}
\newtheorem{definition}{Definition}
\newtheorem{lemma}{Lemma}
\newtheorem{assumption}{Assumption}

\def\BibTeX{{\rm B\kern-.05em{\sc i\kern-.025em b}\kern-.08em
    T\kern-.1667em\lower.7ex\hbox{E}\kern-.125emX}}
\begin{document}
\title{Control
Density Function for Robust Safety and Convergence}
\author{Joseph Moyalan, \IEEEmembership{Student Member, IEEE}, Sriram S. K. S Narayanan, \IEEEmembership{Student Member, IEEE}, and Umesh Vaidya, \IEEEmembership{Member, IEEE}
\thanks{Financial support from the NSF CPS award 1932458 and NSF 2031573 is greatly acknowledged. Joseph Moyalan, Sriram S. K. S Narayanan, and Umesh Vaidya are with the Department of Mechanical Engineering, Clemson University, Clemson, SC; {email: \{jmoyala, sriramk, uvaidya\}@clemson.edu}.}
}

\maketitle

\begin{abstract}
We introduce a novel approach for safe control design based on the density function. A control density function (CDF) is introduced to synthesize a safe controller for a nonlinear dynamic system. The CDF can be viewed as a dual to the control barrier function (CBF), a popular approach used for safe control design. While the safety certificate using the barrier function is based on the notion of invariance, the dual certificate involving the density function has a physical interpretation of occupancy. This occupancy-based physical interpretation is instrumental in providing an analytical construction of density function used for safe control synthesis. The safe control design problem is formulated using the density function as a quadratic programming (QP) problem.
In contrast to the QP proposed for control synthesis using CBF, the proposed CDF-based QP can combine both the safety and convergence conditions to target state into single constraints. 
Further, we consider robustness against uncertainty in system dynamics and the initial condition and provide theoretical results for robust navigation using the CDF. Finally, we present simulation results for safe navigation with single integrator and double-gyre fluid flow-field examples, followed by robust navigation using the bicycle model and autonomous lane-keeping examples. 
\end{abstract}


\section{Introduction}
Navigation is one of the crucial tasks in robotics applications, enabling robots to move efficiently and effectively in obstacle-riddled environments. Applications include vehicle autonomy on structured or unstructured terrain, space exploration, robotics manipulator, agriculture robotics, and underwater exploration \cite{turan2022autonomous,queralta2020collaborative,sun2021review,binbin2021research,arafat2023vision,qiao2023survey,li2023survey,pfrunder2017real}. The problem essentially consists of navigating the system dynamics from some given initial set to a target set while avoiding collision with the unsafe or obstacle set. 
In the real-world use of navigation algorithms, one has to deal with various uncertainties, including uncertainty in system dynamics, sensing, and estimation of the robot's location and the environment. This makes it imperative to develop navigation algorithms that are robust to these sources of uncertainties. Given the significance of this problem, there is a vast body of literature on this topic addressing various aspects of navigation problems. Some of them include avoiding static \cite{minguez2016motion} and dynamic \cite{chen2020navigation} obstacle sets, incorporating uncertainty because of unmodelled dynamics \cite{castaneda2021pointwise,nguyen2021robust}, incorporating uncertainty in sensing and location of the obstacle sets \cite{wang2017safe,xu2017realizing}, and accounting for external disturbances which act as unknown input to the system dynamics \cite{alan2021safe,kolathaya2018input}.

The potential and navigation functions-based methods dominate the various navigation algorithms \cite{khatib1985real,koditschek1990robot}. These functions create a repelling and attractive field near the unsafe and target set, respectively, to guide the system trajectories to the target set while avoiding collision with the unsafe sets \cite{rimon1990exact}. However, even with relatively simple system dynamics, these methods suffer from the existence of local minima and lack of systematic techniques for constructing such functions \cite{warren1990multiple,park2001obstacle}. More recently, control barrier functions have been proposed for safe navigation. The barrier function at its heart relies on the weaker notion of set invariance proposed by Mitio Nagumo to ensure safety \cite{nagumo1942lage,blanchini1999set}. The barrier function is used in control dynamical system setting as control barrier function (CBF). The use of CBF in a control setting is similar to the use of the control Lyapunov function (CLF) for stabilization. For a given barrier function, the safe controller synthesis is formulated as a quadratic program (QP) \cite{ames2016control,ames2019control}. The QP formulation is attractive due to its convex nature and is especially advantageous for solving safe navigation problems for systems with non-trivial dynamics. However, the control synthesis using the barrier function is inherently nonconvex if formulated as a joint search problem over the barrier function and the feedback controller. 

The safe navigation problem can alternatively formulated in the space of density. The density-based approach can be viewed as dual to barrier function-based approach. The dual formulation has the advantage that the co-design problem for jointly finding the controller and the density function for safety and convergence can be written as a convex optimization problem. This convexity in the co-design problem has been exploited for safe controller synthesis in \cite{rantzer2004analysis,prajna2007convex}. Although convex, the infinite-dimensional nature of the optimization problem presents challenges for its finite-dimensional approximation. The finite-dimensional approximation involves polynomial-based parameterization of the optimization variables and the Sum of Squares (SoS)-based computational method for finding the solution \cite{prajna2004nonlinear,moyalan2023data}.

In contrast to the pure optimization-based approach to the safe control synthesis using the density function, the occupancy-based physical interpretation of the density function can also be exploited for safe navigation problems. The measure associated with the density function has a physical interpretation of occupancy. The occupancy-based interpretation has a natural connection with the linear operator theoretic framework involving Perron-Frobenius and Koopman operators \cite{vaidya2010nonlinear,huang2022convex}. This link to the linear operator theory has provided alternate computational methods, not restricted to polynomial basis function, and based on the finite-dimensional data-driven approximation of these linear operators to solve the safe navigation problem \cite{vaidya2018optimal,moyalan2022navigation,vaidya2023data,moyalan2023convex,moyalan2023off}. Still, the co-design problem is computationally challenging for a system with large dimensional state space. The work in \cite{zheng2023safe} circumvents the deficiency of constructing density functions for solving navigation problems in large dimensional state space. In particular, based on the occupancy-based interpretation of the density function, an analytical construction of the density function is proposed to solve the safe navigation problem for systems with simple integrator dynamics. This is a tremendous breakthrough as it allows us to overcome the challenge associated with the construction of the navigation and potential functions. This paper extends the use of analytically constructed density function as proposed in \cite{zheng2023safe} for safe navigation to systems with non-trivial dynamics.

We propose the use of a control density function (CDF) to solve the safe navigation problem. The main contributions of this paper are as follows. The controller design problem for safe navigation is formulated as a QP problem using the CDF. Unlike QP formulated for the safe controller synthesis using CBF, the QP for CDF can combine both safety and convergence to target conditions using a single constraint. We also provide results on robust safety, where the safe controller is designed to be robust against modeling uncertainty and uncertainty in the initial condition. 
Finally, we present simulation results involving the bicycle model and lane-keeping application to verify the main results of this paper. This paper serves as the extended version of \cite{moyalan2024synthesizing2}.
The robust safety and convergence results are new to this paper. The results proving the continuity property and feedback nature of the safe control input are also new to this paper.

The organization of the paper is as follows. In Section \ref{section_preliminaries}, we present preliminaries and notation used in the paper. In Section \ref{Section_cdf}, we present the main results of the paper on the control density function for safe control design followed by Section \ref{Section_robust} on robust safety. The quadratic program based on the control density function for safe control design is presented in Section \ref{section_QPfinite}. Simulation results are presented in Section \ref{section_simulation} followed by conclusion in Section \ref{section_conclusion}. 

\section{Notations and Preliminaries}\label{section_preliminaries}
\subsection{Notations}
Let $\bX\subseteq\mR^n$ denote an open subset or manifold without boundary. Similarly, $\bX_0,\bX_T$, and $\bX_u$ denote the initial, target, and unsafe set, respectively. With no loss of generality, we will assume that the target set $\bX_T$ is at the origin, i.e., $\bX_T=\{0\}$. $\cC^k(\bX)$ denotes the space of all $k$-times continuous differentiable functions. We also denote $\mathcal{B}(\bX)$ to be the Borel $\sigma$-algebra on $\bX$ and ${\cal M}(\bX)$ as the vector space of real-valued measures on ${\cal B}(\bX)$ and $m(\cdot)$ denotes Lebesgue measure. Let $\mathcal{L}_{\infty}(\bX)$ and $\mathcal{L}_1(\bX)$ be the space of essentially bounded and integrable functions on $\bX$ respectively. Also, we represent $\bar{\bX} := \bX\setminus \mathcal{N}_{\eta}$ where $\mathcal{N}_{\eta}$ represents a small neighborhood of $\eta$ radius around $\bX_T$. The notation $\nabla$ denotes $[\frac{\partial}{\partial x_1},\dots,\frac{\partial}{\partial x_n}]^\top$ the column vector.

\subsection{Analytical Construction of Density Function for Safe Navigation}\label{section:density_function}

This paper is about safe control design for control affine dynamical system of the form 
\begin{align}
\dot \bx=\bff(\bx)+\bg(\bx)\bu\label{eq:drift_syst1}
\end{align}
where $\bx\in \bX\subseteq \mR^n, \bu\in \mR^m$. The vector field $\bff$ and $\bg=(\bg_1,\ldots,\bg_m)$ are assumed to be atleast $\cC^1(\bX)$. Following is the definition of a weaker almost everywhere (a.e.) navigation problem.  
\begin{definition}\label{problema.e.navigation}(Almost everywhere (a.e.) safe navigation) The a.e. safe navigation problem consists of steering the trajectories of system (\ref{eq:drift_syst1}) starting from almost all (w.r.t. Lebesgue measure) initial conditions from the initial set $\bX_0$ to the target set $\bX_T$ while avoiding the unsafe set $\bX_u$. 
\end{definition}
In \cite{zheng2023safe}, analytical construction of density function, $\rho(\bx)$ for solving the a.e. safe navigation problem for a single integrator is provided. In particular, with integrator system dynamics of the form $ \dot \bx=\bu$, a feedback control input as the positive gradient of the density function, i.e., $\bu=\nabla \rho(\bx)^\top$ is proposed for solving a.e. safe navigation problem as defined in Definition~\ref{problema.e.navigation}. In this paper, one of the main contributions is extending the use of the density function for solving a.e. navigation problem for nonlinear systems with nonzero drift as given in  (\ref{eq:drift_syst1}). In the following discussion, we provide the details of the analytically constructed density function as proposed in \cite{zheng2023safe}.

Let there be $K$ number of obstacles. The obstacle sets are characterized using a continuous function, $c_k(\bx)$, whose level sets form connected components. In particular, we define
\begin{align}
    \bX_{u_k} := \{\bx\in \bX: c_k(\bx) \le 0\} \label{eq:obs_set}
\end{align}
Hence, the unsafe region is given by $\bX_u := \bigcup_{k=1}^K \;\bX_{u_k}$. Next, we define a sensing region, $\bX_{b_k}$, around the unsafe set $\bX_{u_k}$  using a scalar-valued function $b_k(\bx)$, which is again assumed to be continuous with level sets forming connected components.
\begin{align}
    \bX_{b_k} := \{\bx \in \bX:b_k(\bx) \le 0 \} \setminus \bX_{u_k}
    \label{eq:sensing_set}
\end{align}
Now, we will define a smooth inverse bump function by making use of $c_k(\bx)$ and $b_k(\bx)$. We first start by constructing the following functions,
\begin{align*}
    m_k(\bx) &= \frac{c_k(\bx)}{c_k(\bx)-b_k(\bx)},\\
    \psi_k(\bx) &= \frac{\exp(\frac{-1}{m_k(\bx)})}{\exp(\frac{-1}{m_k(\bx)})+\exp(\frac{-1}{1-m_k(\bx)})}.
\end{align*}
Using the functions $m_k(\bx)$ and $\psi_k(\bx)$, we define a piece-wise smooth inverse bump function $\Psi_k(\bx)$ as follows:
\begin{align}\label{eq:bump}
    &\Psi_k(\bx) = \begin{cases}
        0, &\bx \in \bX_{u_k} \\
        \psi_k(\bx), & \bx \in \bX_{b_k} \\
        1, & \text{otherwise}
    \end{cases}.
\end{align}
Note that $\bPsi(\bx)=\prod_{k=1}^K\Psi_k(\bx)$ encodes the information of the unsafe set $\bX_{u}$. Next,  we will use a distance function $D(\bx)$, utilizing the current state and the target state, to incorporate the information of the target set $\bX_T$ in the density function. The $D(\bx)$ can be modified to adjust to the geometry of the underlying configuration space. In the Euclidean space with $\bx \in \mathds{R}^n$, we use $D(\bx) = \bx^\top P \bx$, for some $P>0$, where the target set is assumed to be at the origin.  Therefore, the density function $\rho(\bx)$ for the safe control is given as follows:
\begin{eqnarray}
    \rho(\bx) = \frac{\bPsi(\bx)}{D(\bx)^{\alpha}} \label{eq:rho_S}
\end{eqnarray}
for some $\alpha >0$. The following theorem from \cite{zheng2023safe} provides the solution for a.e. safe navigation for $\dot{\bx}=\bu$.
\begin{theorem}[\cite{zheng2023safe}]
    Consider $\rho(\bx)$ given in \eqref{eq:rho_S}. Then the safe control given by $\bu(\bx)=\nabla \rho(\bx)$ is the solution to the a.e. safe navigation as stated in Definition~\ref{problema.e.navigation} for $\dot{\bx}=\bu$.
\end{theorem}
Fig. \ref{fig:density_function}a shows an environment with one obstacle set $\bX_u$ and target $\bX_T$. Fig. \ref{fig:density_function}b shows the corresponding density function representation. Note that the density function $\rho(\bx)$ takes minimum value for $\bx \in \bX_u$ and maximum value for $\bx \in \bX_T$.
\begin{figure}[ht]
  \centering
  \includegraphics[width=1\linewidth]{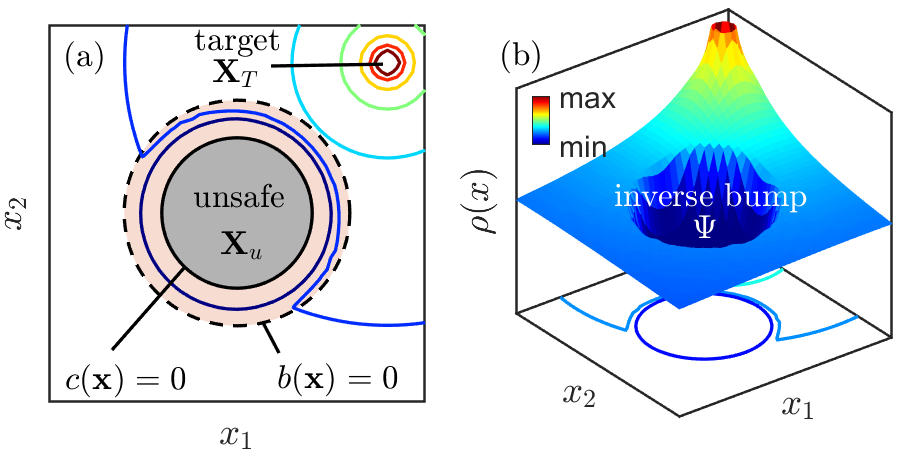}
\caption{(a) Environment setup with unsafe set $\bX_u$ and target $\bX_T$, (b) density function $\rho(\bx)$ for this environment.}
\label{fig:density_function}
\end{figure}
\section{Control Density Function (CDF) for Safety}\label{Section_cdf}

This section introduces the CDF for solving motion planning problems with safety constraints. This will be along the lines of using CLF \cite{artstein1983stabilization,sontag1983lyapunov} and CBF \cite{wieland2007constructive,ames2019control} for stabilization and safe navigation, respectively. However, unlike the Lyapunov and barrier functions, the safety and convergence requirements can be combined using a single density function. 

The safe navigation problem using the density function approach is, in fact, a co-design problem of jointly finding a feedback controller $k(\bx)$ and a density function $\rho(\bx)$. One of the main advantages of using the density function over the barrier function as a safety certificate is that the co-design problem of jointly finding the density function for safety and the feedback controller is convex. This convexity property is exploited in the work of \cite{moyalan2023convex,yu2022data} to design the safety controller. Although the co-design problem is convex, solving this convex problem involves parameterization of the density function and controller using some basis functions. Determining appropriate parameterization is challenging, especially for systems with large dimensional state space. In this paper, we take a different approach where we make the use of analytically constructed density function as proposed in \cite{zheng2023safe} (Eq. \ref{eq:rho_S}) to determine the safety controller for the system with non-trivial drift terms. The main contribution is to show that the a.e. safe navigation problem as defined in Definition \ref{problema.e.navigation} can be formulated as a quadratic programming problem using the density function given in (\ref{eq:rho_S}). 
Following is the first main result of this paper.

\begin{theorem}\label{thm:11} Let the density function $\rho(\bx)$ be as given in (\ref{eq:rho_S}).
The a.e. navigation problem as stated in Definition \ref{problema.e.navigation} is solvable for the control system (\ref{eq:drift_syst1}) if there exists a controller $\bu=\bk(\bx)$ and $\lambda>0$ satisfying the following conditions
\begin{subequations}
    \begin{align}
&\nabla\cdot((\bff +\bg \bk)\rho)\geq 0,\;\;a.e.\;\bx\in \bar \bX\label{condition1}\\
&\nabla\cdot((\bff +\bg \bk)\rho)\geq \lambda>0,\;\;\;\bx\in \bX_o\label{condition2}
\end{align}
\end{subequations}

\end{theorem}

We postpone the proof of this theorem to the Appendix.

The inequalities in (\ref{condition1})-(\ref{condition2}) are linear in $\bu=\bk(\bx)$ for a given density function $\rho(\bx)$. As stated above, solving for $\rho$ and $\bu=\bk(\bx)$ can also be written as an infinite-dimensional convex optimization problem as follows. Let $\bar \brho(\bx):=\rho(\bx)\bk(\bx)$, then the inequalities (\ref{condition1})-(\ref{condition2}) can be written as following linear inequalities in terms of variables $\rho$ and $\bar \brho$
\begin{subequations}
    \begin{align}
&\nabla\cdot(\bff\rho  +\bg \bar \brho)\geq 0,\;\;a.e.\;\bx\in \bar \bX\\
&\nabla\cdot(\bff\rho +\bg \bar\brho)\geq \lambda>0,\;\;\;\bx\in \bX_o
\end{align}
\end{subequations}
 
Although convex, the inequalities are infinite-dimensional and, hence, computationally challenging to solve for global safety controllers. Instead, if we solve for $\bu(\bx)$ along a trajectory for a given initial condition, then the problem of finding the control input $\bu(\bx)$ can be written as a finite-dimensional convex optimization problem. In fact, similar to the quadratic programming (QP) based formulation of safety using the CBF, safe navigation using the CDF can also be formulated as a QP problem. For any given fixed $\bx$, the QP for CDF is written as follows.
\begin{subequations}\label{eq:qp-cdf}
    \begin{align}
        &\min_{\bu}\; \|\bu\|^2  \\
       \text{s.t.}\;\;\;& \nabla \cdot (\bff (\bx)\rho + \bg(\bx)\bu\rho) \ge 0,\;\;a.e.\;\bx\in\bar{\bX}\\ 
    & \nabla \cdot (\bff (\bx)\rho + \bg(\bx)\bu\rho) \ge \lambda > 0,\;\;\forall\;\bx\in\bX_0
\end{align}
\end{subequations}


One of the advantages of the QP formulation given in (\ref{eq:qp-cdf}) is that there is no need for the reference control input as the convergence information with respect to the target set is encoded in the construction of the analytical density function. However, if there exists a nominal control, $\bu_0$, for the system given by \eqref{eq:drift_syst1}, we can reformulate \eqref{eq:qp-cdf} to enforce $\bu_0$ in the absence of unsafe sets by introducing a cost function which minimizes $\|\bu - \bu_0\|^2$, i.e., 
\begin{subequations}\label{eq:qp-cdf1}
    \begin{align}
        &\min_{\bu}\; \|\bu-\bu_0\|^2  \\
       \text{s.t.}\;\;\;& \nabla \cdot (\bff (\bx)\rho + \bg(\bx)\bu\rho) \ge 0,\;\;a.e.\;\bx\in\bar{\bX}\\ 
    & \nabla \cdot (\bff (\bx)\rho + \bg(\bx)\bu\rho) \ge \lambda > 0,\;\;\forall\;\bx\in\bX_0 
    \end{align}
\end{subequations}



We write the QP in the more general form for our next results which are as follows.
\begin{subequations}\label{QP-CDF_base}
\begin{align}
    &\bu^{\star}(\bx) = \argmin_{\bu} \;\;\;\;\bu^\top \bH(\bx) \bu+ \bJ(\bx)\bu\\
    &\text{s.t.}\;\;\;\nabla \cdot (\bff(\bx)\rho + \bg(\bx)\rho\bu)\ge 0,\;\;a.e.\;\bx\in\bar{\bX} \label{constraint_1}\\
    & \;\;\;\;\;\;\;\nabla \cdot (\bff (\bx)\rho + \bg(\bx)\bu\rho) \ge \lambda > 0,\;\;\forall\;\bx\in\bX_0\label{constraint_2}
\end{align}
\end{subequations}

Here, $\bu = [u_1,\dots,u_m]^\top, \bH(\bx) \in \mathds{R}^{m \times m}$ is some positive matrix for fixed $\bx$, and $\bJ(\bx) \in \mathds{R}^{m}$. The next theorem proves that the solution of the QP problem will lead to a feedback controller satisfying $\mathcal{C}^1$ continuity property.

\begin{theorem}\label{theorem_statefeedback}
Assume that the vector field $\bff$, $\bg$, and the matrix valued function $\bH(\bx)$ and $\bJ(\bx)$ in (\ref{QP-CDF_base}) are atleast $\cC^1$ function of $\bx$.
Then the solution $\bu^{\star}(\bx)$ obtained by solving \eqref{QP-CDF_base} is $\mathcal{C}^1$ continuous for almost all $\bx \in \bar{\bX}$ and is feedback.
\end{theorem}
\[\]
We postpone the proof of this theorem to the Appendix. 
One of the main differences between the QP for CDF vs. QP for CBF is that the constraints in the QP-CDF involve the spatial derivative of the control input term. In  Section \ref{section_QPfinite}, we present an approach to approximate the spatial derivative.

 \section{Robust safety and convergence}\label{Section_robust}

In this section, we propose using the control density function for robust safety. We consider robustness against uncertainty in system dynamics and the initial condition. 

\subsection{Uncertainty in System Dynamics}
Consider an uncertain dynamical system of the form 
 \begin{align}
    \dot{\bx} = \bff(\bx)  + \bg(\bx)\bu +  \bff_\delta(\bx,t) \label{eq:sys_dyn_uncertainty}
\end{align}
where $\bff_\delta$ denotes the perturbation term, possibly time-varying, modeling the uncertainty in the system dynamics. The following assumptions are made on the uncertain term, $\bff_\delta$. 
\begin{assumption}\label{assume_perturbation}
We assume that the perturbation satisfies the following bounds 
\begin{align}
\|\bff_{\delta}(\bx,t)\|\leq c_{\delta_1},\;\;\;|\nabla\cdot(\bff_\delta(\bx,t))|\leq c_{\delta_2} ,\;\;\;\bx\in \bX,\;\;\;\forall t\geq 0. 
\end{align}
for some positive constants $c_{\delta_1}$ and $c_{\delta_2}$. 
\end{assumption}
\begin{assumption}\label{assume:3}
    For $\bx \in \bar \bX$, we assume the following bounds

\begin{align}
\left\|\frac{\nabla D(\bx)}{D(\bx)}\right\|\leq c_{\partial D},\;\;\left\|\frac{\partial \bPsi(\bx)}{\partial \bx}\right\|\leq c_{\Psi} \bPsi(\bx), 
\end{align}
for some positive constant $c_{\partial D}$ and $c_\Psi$. Note that outside the transition region, we have $\frac{\partial \bPsi}{\partial \bx}=0$. 
\end{assumption}

To prove the results on robust safety, we first need to prove the following Lemma for a more general time-varying system of the form 
\begin{align}
\dot \bx=\bF(\bx,t)\label{time_varying}
\end{align}
where $\bx\in \bX$ and $\bF$ is assumed to be atleast $\cC^1$ function of $\bx$ for any fixed $t$. 
Unlike Theorem~\ref{thm:11} and its proof, which are based on the occupancy-based interpretation of the density function, the following Lemma is proved using Liouville's theorem \cite{masubuchi2021lyapunov,rantzer2004analysis} for the time-varying systems which involve utilizing the Liouville equation in the integral form used for the evolution of densities.

\begin{lemma}\label{thm:1} Consider the dynamical system (\ref{time_varying}) with density function as given in (\ref{eq:rho_S}). For all $t \ge 0$, if following inequalities are satisfied 
\begin{subequations}\label{theorem_inequalities}
    \begin{align}
&\nabla\cdot(\bF(\bx,t)\rho)\geq 0,\;\;a.e.\;\bx\in \bar \bX\\
&\nabla\cdot(\bF(\bx,t)\rho)\geq \lambda>0,\;\;\;\bx\in \bX_0
\end{align}
\end{subequations}

then almost every initial condition from $\bX_0$ will be steered to the target set $\bX_T$ while avoiding the unsafe set $\bX_u$.
\end{lemma}
The proof of this Lemma is differed to the Appendix. 
Following is the main result of this paper on robust safety with uncertainty in the system dynamics.

\begin{theorem} Consider the uncertain dynamical system given in \eqref{eq:sys_dyn_uncertainty} satisfying Assumptions \ref{assume_perturbation} and \ref{assume:3}. If there exists  $\bu=\bk(\bx)$ satisfying 
\begin{subequations}\label{eq:robust_qp_cdf}
    \begin{align}
&\nabla\cdot(\bff(\bx) \rho+\bg(\bx) \bk(\bx)\rho)\geq \gamma \rho,\;\;a.e.\;\bx\in \bar \bX\label{eq1}\\
&\nabla\cdot(\bff(\bx) \rho+\bg(\bx) \bk(\bx)\rho)\geq \lambda+\gamma \rho > 0,\;\;\;\bx\in \bX_0\label{eq2}
\end{align}
\end{subequations}

where $\gamma=c_{\delta_2}+\alpha c_{\delta_1} c_{\partial D}+c_{\delta_1}c_{\Psi}$, then almost every initial condition from $\bX_0$ will be steered to the target set $\bX_T$ while avoiding the unsafe set $\bX_u$. 
\end{theorem}

\begin{proof}
We need to show that (\ref{eq1})-(\ref{eq2}) implies the following inequality in (\ref{ss})-(\ref{ss1}), the proof then will follow using the results of Lemma \ref{thm:1}.
\begin{subequations}
    \begin{align}
\nabla \cdot((\bff(\bx)+\bg(\bx)\bk(\bx)+\bff_\delta(\bx,t))\rho)&\geq 0, \nonumber \\\;\;a.e.\; \bx\in \bar \bX,\;\forall t \ge 0\label{ss}\\
\nabla \cdot((\bff(\bx)+\bg(\bx)\bk(\bx)+\bff_\delta(\bx,t))\rho)&\geq \lambda>0,\nonumber \\ \bx\in \bX_0,\;\forall t \ge 0\label{ss1}
\end{align}
\end{subequations}
Defining $\bF_c(\bx):=\bff(\bx)+\bg(\bx)\bk(\bx)$, we obtain
\[\nabla\cdot(\bF_c(\bx)\rho+\bff_\delta(\bx,t) \rho)=\nabla\cdot(\bF_c(\bx)\rho)+\nabla\cdot(\bff_\delta(\bx,t)\rho).\]
We have the following estimates for $\nabla\cdot (\bff_\delta(\bx,t) \rho)$.
\[|\nabla\cdot (\bff_\delta(\bx,t) \rho)|\le|\nabla\cdot \bff_\delta(\bx,t)| \rho+|(\nabla \rho)\cdot \bff_\delta(\bx,t)|\]
as $\rho$ is positive. Following Assumption \ref{assume_perturbation}, we have
\begin{align}
    |\nabla\cdot \bff_\delta(\bx,t)|\rho\leq c_{\delta_2}\rho \label{eq:proof_11}
\end{align}
Now,
\[(\nabla\rho)\cdot \bff_\delta(\bx,t)= -\alpha \frac{\nabla D}{D} \bff_\delta(\bx,t) \rho+\frac{1}{D^\alpha} \frac{\partial \Psi}{\partial \bx}\bff_\delta(\bx,t)\]
where $\rho(\bx) = \frac{\bPsi(\bx)}{D(\bx)^{\alpha}}$ as given in \eqref{eq:rho_S}.
Hence,
\begin{align}|(\nabla\rho)\cdot \bff_\delta(\bx,t)|\leq \alpha \left\|\frac{\nabla D}{D}\right\|\left\| \bff_\delta(\bx,t) \right\| \rho+\nonumber\\
\frac{1}{D^\alpha} \left\|\frac{\partial \Psi}{\partial \bx}\right\| \left\|\bff_\delta(\bx,t)\right\|\label{eq3}
\end{align}
Using the bounds from Assumptions (\ref{assume_perturbation}) and (\ref{assume:3}), we obtain
\begin{align}|(\nabla\rho)\cdot \bff_\delta(\bx,t)|\leq (\alpha c_{\delta_1} c_{\partial D}+c_{\delta_1}c_{\Psi}) \rho\label{eq4}
\end{align}
Combining \eqref{eq:proof_11} and (\ref{eq4}), we obtain
\[|\nabla\cdot (\bff_\delta(\bx,t) \rho)|\leq \gamma \rho\]
where $\gamma:=c_{\delta_2}+\alpha c_{\delta_1} c_{\partial D}+c_{\delta_1}c_{\Psi}$. Now
\[\nabla\cdot((\bF_c(\bx)+\bff_\delta(\bx,t) )\rho)\geq \nabla\cdot(\bF_c(\bx)\rho)-|\nabla\cdot (\bff_\delta(\bx,t) \rho)| \geq 0\]
where the above follows from inequality (\ref{eq1}). Similarly, inequality (\ref{ss1}) follows from inequality (\ref{eq2}).
\end{proof}

\subsection{Uncertainty in System States}

In this section, we consider uncertainty in the state of the system. The uncertainty in the state estimate could be due to sensor noise or the error in the state estimate. For example, in robotic applications, uncertainty can arise due to uncertainty in the localization algorithm used to localize the robot. One could also consider uncertainty in the location of the unsafe set. However, this can typically be handled by enlarging the unsafe set to include the uncertainty bounds. We will assume that the uncertainty in the state estimate is bounded. Let $\bx_0(t)$ be the nominal or estimated state of the system at time $t$, and we assume that the true state, $\bx(t)$ of the system can be anywhere within $\beta$ neighborhood of $\bx_0(t)$, i.e., 
\begin{align}
\|\bx(t)-\bx_0(t)\|\leq \beta
\end{align}
Let ${\cal N}_\beta(\bx_0)$ denotes the $\beta$ neighborhood of $\bx_0$.
The safe navigation problem with the uncertainty in the system state can then be written as 
\begin{subequations}\label{infinite_robust}
    \begin{align}
\min_{\bu}\;\;&\|\bu\|^2\\
\max_{\bx\in {\cal N}_\beta (\bx_0)}\;\;&-\nabla\cdot ((\bff+\bg\bu)\rho)\leq 0,\;\;a.e.\;\bx\in \bar \bX\\
&-\nabla\cdot ((\bff+\bg\bu)\rho)+\lambda\leq 0,\;\;\bx\in \bX_0
\end{align}
\end{subequations}

The above problem is an infinite-dimensional robust optimization problem. The goal is to minimize the control effort for each $\bx$, but the robustness involves maximizing the constraints for all $\bx\in {\cal N}_\beta(\bx_0)$. 
Unlike non-robust safety problems, which can be written as finite-dimensional QP, if we solve for $\bu$ along the trajectory, the above problem is
infinite-dimensional and nonconvex even if we solve for $\bu$ along a trajectory. This is due to the robustness $\max$ constraints. The above problem can be converted to a finite-dimensional convex optimization problem by sampling over the set ${\cal N}_\beta(\bx_0)$. Let $\{\bx_p\}_{p=1}^N$ be the data points sampled randomly from the set ${\cal N}_\beta(\bx_0)$, we can then approximate the infinite-dimensional constraints with finitely many constraints as 
\begin{subequations}\label{finite_robust}
  \begin{align}
&\min_{\bu}\;\;\|\bu\|^2\\
&{\rm s.t.}\;\;\nabla\cdot((\bff(\bx_p)+\bg(\bx_p) \bu)\rho(\bx_p))\geq 0,\;\;p=0,1,\ldots, N \label{finite_robust_1}
\end{align}  
\end{subequations}

Since the data points $\bx_p$ are drawn randomly, the solution to the finite-dimensional optimization problem (\ref{finite_robust}) will be random in nature. In the following discussion, we will assume that the sampling probability is uniform w.r.t. Lebesgue measure. 
Now we can use results from scenario optimization from \cite{calafiore2006scenario} to connect the results obtained using the finite-dimensional optimization problem (\ref{finite_robust}) to the original infinite dimensional nonconvex robust optimization problem (\ref{infinite_robust}). Toward this goal, we define 
\[{\cal F}(\bx,\bu):=-\nabla\cdot((\bff +\bg \bu)\rho)\]
and the probability of violation as
\[G(\bu):={\rm Prob}\{\bx\in {\cal N}_{\beta}(\bx_0): {\cal F}(\bx,\bu)>0\}\]
Hence, $G(\bu)$ measures the volume of bad parameter $\bx$ for which the constraints ${\cal F}(\bx,\bu)\leq 0$ are violated. One can define $\epsilon$-level feasible solution as any solution for which $G(\bu)\leq \epsilon$. 
\begin{definition}($\epsilon$-Level Solution): Let $\epsilon\in (0,1)$. We say that $\bu\in \bU$ is 
is an $\epsilon$-level robustly feasible (or, more simply, an $\epsilon$-level)
solution, if $G(\bu)\leq \epsilon$.
\end{definition}
Since the samples $\bx_p$ are drawn randomly, the optimal solution, $\hat \bu_N$ of the optimization problem (\ref{finite_robust}) will be $\epsilon$-level solution for the given random extraction and not the other. The $\sigma$ is used to bound the probability that $\hat \bu_N$ is not
a $\epsilon$-level solution for other random extraction. Hence, the parameter $\sigma$ measures the risk of failure, or confidence,
associated with the randomized solution algorithm. In particular, the following theorem follows from \cite{calafiore2006scenario} connecting the optimal solution of the optimization problem (\ref{infinite_robust}) to that of (\ref{finite_robust}). 
\begin{theorem} Assume that all possible realization of samples $\{\bx_p\}_{p=1}^N$ either lead to an infeasible solution for the optimization problem (\ref{finite_robust}) or, if feasible, admits a unique optimal solution. Fix two real numbers $\sigma\in (0,1)$ (confidence parameter) and $\epsilon$ (level parameter). If 
\[N\geq \Big{\lceil}\frac{2}{\epsilon}\ln \frac{1}{\sigma}+2m +\frac{2m}{\epsilon} \ln \frac{2}{\epsilon}\Big{\rceil}\]
where $m$ is the number of control inputs and $\lceil\cdot \rceil$ denotes the smallest integer greater than or equal to the argument. Then with probability no smaller than $(1-\beta)$, either (\ref{finite_robust}) is unfeasible or (\ref{finite_robust}) is feasible, and then its optimal solution, $\hat \bu_N$, is $\epsilon$-level robust feasible.  
\end{theorem}

\section{Algorithm for QP-CDF}\label{section_QPfinite}
This section presents the algorithm for solving the QP presented in  Eq. (\ref{eq:qp-cdf1}). In particular, the main part of the algorithm is the approximation of the spatial derivative of the control input $\bu$.
The constraints in the QP can be rewritten for the multi-input case as 
\begin{align}    \nabla\cdot(\bff(\bx)\rho)+\sum_{j=1}^m \nabla\cdot (\bg_j(\bx)\rho)u_j + \sum_{j=1}^m \nabla u_j \bg_j \rho \geq 0 \label{eq:expansion}
\end{align}
where $\bg=(\bg_1,\ldots, \bg_m)$ and $\nabla u_j$ is a row vector. 
We split the inequality in \eqref{eq:expansion} as follows:
\begin{subequations}
    \begin{align}
    \nabla \cdot (\bff\rho) + \sum_{j=1}^m \nabla\cdot (\bg_j\rho)u_j &\geq \zeta \label{eq:inequality_split1}\\
\sum_{j=1}^m \nabla u_j \bg_j \rho &\geq -\zeta   \label{eq:inequality_split}    
\end{align}
\end{subequations}

for some $\zeta>0$. 
The following explains the procedure for approximating the spatial gradient involving the control input in the optimization problem (\ref{eq:qp-cdf1}).  The left-hand side of the (\ref{eq:inequality_split}) is equal to the Lie derivative of the $u_j$ along the direction of the vector field $\rho\bg_j$ denoted by $\cL_{\rho \bg_j} u$. Let $\phi_{\rho\bg_j}^{\Delta t}(\bx)$ be the flow of vector field $\rho \bg_j$ for time step $\Delta t$. For $\Delta t$ sufficiently small, we have
\begin{align}
\cL_{\rho\bg_j} u_j \approx \frac{u_j(\phi_{\rho\bg_j}^{\Delta t}(\bx))-u_j(\bx)}{\Delta t}
\end{align}
Furthermore, for sufficiently small $\Delta t$, we can approximate 
\[\phi_{\rho\bg_j}^{\Delta t}(\bx)\approx \bx+\Delta t \rho(\bx)\bg_j(\bx).\]
Letting $\bz_j=\bx+\Delta t\rho(\bx) \bg_j(\bx)$, then (\ref{eq:inequality_split}) can be approximated as 
\begin{align}
\sum_{j=1}^m u_j(\bz_j)-u_j(\bx)\geq -\zeta\Delta t 
\end{align}
 Note that while we are interested in computing the control value at point $\bx$, i.e., $\bu(\bx)$, due to the spatial gradient of $\bu$ in the optimization problem, we are required to know the value of $u_j$ at $\bz_j$. Hence, in the optimization problem, to solve for $\bu(\bx)$, we have to introduce an additional variable $u_j(\bz_j)=:\bar u_j$. Now, let $\bar \bu:=(\bar u_1,\ldots,\bar u_m)^\top$, $\bar \bu_0=(u_{10}(\bz_1),\ldots, u_{m0}(\bz_m))^\top$, and $\bu_0=(u_{10},\ldots, u_{m0})^\top$. Here, $ \bu_0$ and $\bar \bu_0$ are the value of nominal input at $\bx$ and $\{\bz_1,\dots,\bz_m\}$ respectively. Therefore, we can write the optimization problem as follows:
\begin{subequations}
    \begin{align}
&\min_{\bu,\bar \bu,\zeta} \|\bu(\bx)-\bu_0(\bx)\|+\|\bar \bu-\bar \bu_0\|+\zeta^2\\
&\nabla\cdot(\bff(\bx)\rho(\bx))+\sum_{j=1}^m \nabla\cdot(\bg_j(\bx) \rho(\bx)) u_i(\bx)\geq \zeta \\
&\nabla\cdot(\bff(\bz_j)\rho(\bz_j))+\sum_{i\neq 
 j}^m \nabla\cdot(\bg_i(\bz_j) \rho(\bz_j)) u_{i0}(\bz_j)+\nonumber\\&\nabla\cdot (\bg_j(\bz_j) \rho(\bz_j))\bar u_j\geq \zeta, \\
&\sum_{j=1}^m \bar u_j- u_j(\bx)\geq -\zeta\Delta t  
\end{align}
\end{subequations}

where $j=1,\ldots,m$ and $\bz_j=\bx+\Delta t \rho(\bx)\bg_j(\bx)$.
The above equation has $2m+1$ decision variables: $\bu(\bx), \bar \bu=(u_1(\bz_1),\ldots, u_m(\bz_m))$ and $\zeta$.   
\begin{algorithm}\label{algo_1}
\caption{QP-CDF}
\textbf{Input:} $\bff,\bg,\rho,\bx_0,\bu_0,\bar{\bu}_0,N$\\
\For{$k=1:N$}{
$\bz_j=\bx_{k-1}+\Delta t\rho(\bx_{k-1}) \bg_j(\bx_{k-1})\;\;\forall\;j=1,\dots,m$.\\
    \textbf{Solve} \For{$\bu_k,\bar{\bu}_k,\zeta$}{
    $\min \;\| \bu_k - \bu_{0k}\|^2+\|\bar{\bu}_k - \bar{\bu}_{0k}\|^2 + \zeta^2$\\
    s.t. \\
    $\nabla \cdot (\bff(\bx_{k-1}) \rho(\bx_{k-1}))+$\\
    \quad \quad \quad \quad $\sum_{j=1}^m \nabla\cdot(\bg_j(\bx_{k-1}) \rho(\bx_{k-1})) u_{ik}(\bx_{k-1})\geq \zeta$,\\
    \quad \\
    $\nabla \cdot (\bff(\bz_1) \rho(\bz_1))+$\\
    \quad \quad \quad \quad $\sum_{i\neq 1}^m \nabla\cdot(\bg_i(\bz_1) \rho(\bz_1)) u_{i0k}(\bz_1) +$,\\
    \quad \quad \quad \quad $\nabla\cdot (\bg_1(\bz_1) \rho(\bz_1))\bar u_{1k}\geq \zeta$\\
    \quad \\
    $\;\;\;\;\;\;\;\vdots$ \\
    \quad \\
    $\nabla \cdot (\bff(\bz_m) \rho(\bz_m))+$\\
    \quad \quad \quad \quad $\sum_{i\neq m}^m \nabla\cdot(\bg_i(\bz_m) \rho(\bz_m)) u_{i0k}(\bz_m) +$,\\
    \quad \quad \quad \quad $\nabla\cdot (\bg_m(\bz_m) \rho(\bz_m))\bar u_{mk}\geq \zeta$\\
    \quad \\
     $\sum_{j=1}^m \bar u_j- u_j(\bx)\geq -\zeta\Delta t$ 
    }
$\bx_k = \bx_{k-1} + \Delta t (\bff(\bx_{k-1}) + \bg(\bx_{k-1})\bu_k)$
}
\end{algorithm}

\section{Simulation 
Results}\label{section_simulation}

In this section, we provide some examples to showcase the power of CDF for safe navigation. We start with the single integrator and double-gyre flow-field dynamics to demonstrate the application of QP-CDF for obstacle avoidance. Then, we validate the robust applications of QP-CDF by utilizing it for safe navigation for bicycle models and lane-keeping examples. 
\subsection{Comparison of CDF with CBF}
Consider the single integrator dynamics as follows:
\begin{subequations}\label{eq:single_integrator}
   \begin{align}
    \dot{x}_1 &= u_1 \\
    \dot{x}_2 &= u_2 
\end{align} 
\end{subequations}

where $u_1$ and $u_2$ are velocity controls to position states $x_1$ and $x_2$ respectively. The safe control for \eqref{eq:single_integrator} is obtained by solving the QP-CDF given in \eqref{eq:qp-cdf1}. In Fig.~\ref{fig:cbf_cdf_comparison}, we provide a single obstacle example with multiple sensing regions. The center of the obstacle is located at $[0,0]$, and its radius is 1 unit. The goal is to move from the initial point given by $\bx_0 = [-5,0]$ to the target point at $\bx_T = [5,0]$ while avoiding the circular obstacle. The radius of the different sensing regions is given by $b_j = \{2,\;3,\;4\}$. In Fig.~\ref{fig:cbf_cdf_comparison}, we provide the trajectories obtained through safe navigation control by solving the QP-CDF and comparing it with QP-CBF for different degrees of safety marked by the sensing regions. The general formulation of QP-CBF to obtain the safe control for the control-affine dynamical system as given in \eqref{eq:drift_syst1} is given below:
\begin{subequations}\label{eq:qp-cbf}
    \begin{align}
    &\min_{\bu}\; \|\bu\|^2  \\
       \text{s.t.}\;\;\;& \dot{h}(\bx,\bu) \ge -e_1 h(\bx) \\
       & \dot{V}(\bx,\bu) \le -e_2\; V(\bx) 
\end{align}
\end{subequations}

for some $e_1,\;e_2 > 0$. Here, $\bu_0$ is the reference control, $h(\bx) = 1- \|\bx\|^2$ represents the barrier function incorporating the safety constraints, and $V(\bx) = \|\bx-\bx_T\|^2$ represents the Lyapunov function incorporating the convergence constraints in the QP-CBF. The parameters $e_1$ and $e_2$ are used to tune the degree of safety and the convergence rate for the QP-CBF, respectively. In contrast, QP-CDF only needs one constraint involving density function $\rho(\bx)$ to enforce both safety and convergence. We use $e_1 = \{0.3,\;0.5,\;0.7\}$ and the same value of $e_2=0.5$ for all trajectories in the given example. Due to the exponential nature of the CBF constraint, the degree of safety drops drastically between $e_1 = 0.3$ and $e_1 = 0.5$ (see dashed blue vs dashed purple line). However, for the QP-CDF, the variation in the degree of safety changes more uniformly based on the sensing region $b_j$. Fig.~\ref{fig:cbf_cdf_control_plots} compares the control plots obtained from QP-CDF and QP-CBF for the different degrees of safety. Therefore, for complex environments with multiple obstacle sets, a desired level of safety can be achieved easily using QP-CDF.
\begin{figure}[ht]
  \centering
  \includegraphics[width=0.66\linewidth]{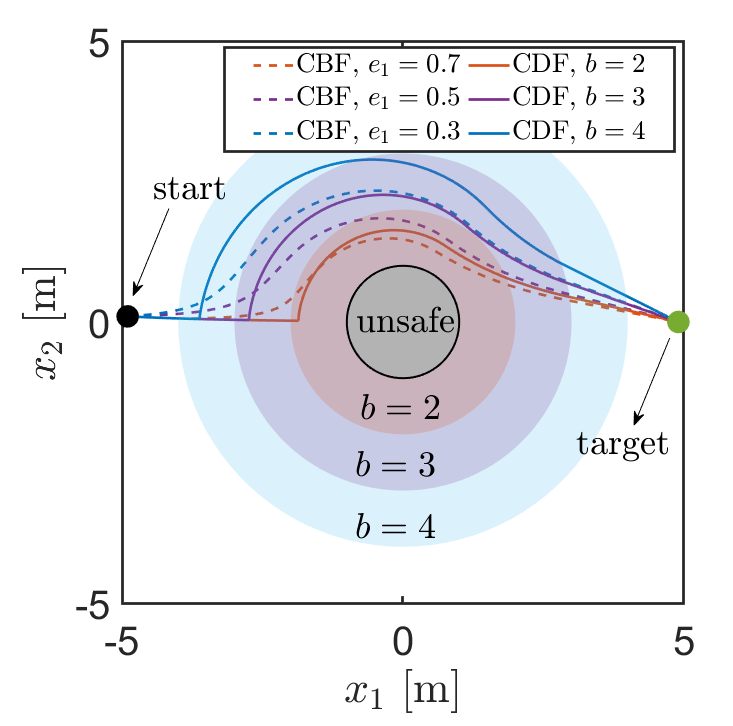}
\caption{\textbf{Comparison with CBFs:} Performance of QP-CBF (dashed lines) with varying $e_1$ compared with QP-CDF (solid lines) with varying $b_j$ for single integrator dynamics.}
\label{fig:cbf_cdf_comparison}
\end{figure}

\begin{figure}[ht]
  \centering
  \includegraphics[width=1\linewidth]{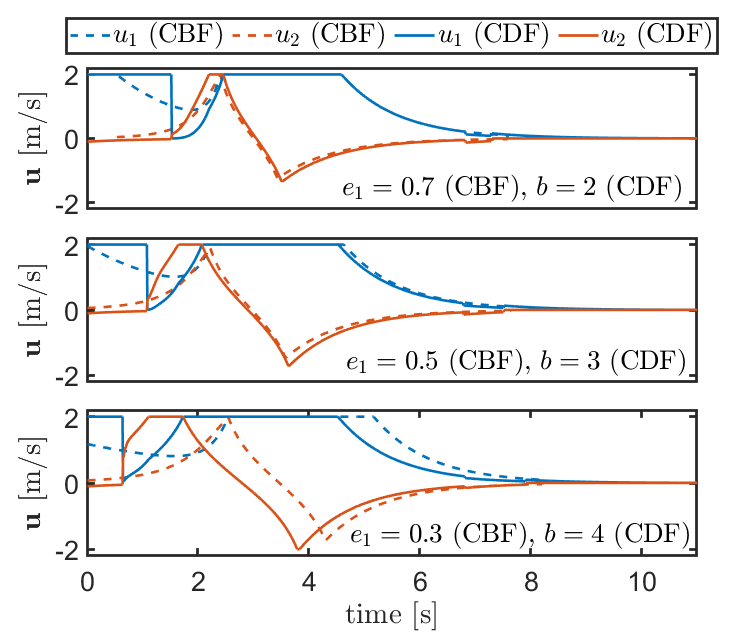}
\caption{\textbf{Comparison with CBFs:} Control plots for QP-CBF (dashed lines) with varying $e_1$ compared with QP-CDF (solid lines) with varying $b_j$ for single integrator dynamics.}
\label{fig:cbf_cdf_control_plots}
\end{figure}
\subsection{Double-Gyre Flow-field}
Our second example is for navigating a robot under the influence of a double-gyre flow field. The double-gyre vector field is used as a model for oceanographic flows, and hence, this example can be considered as a model for underwater robots navigating in oceanographic flow \cite{knizhnik2022flow}.  
\begin{subequations}\label{double_gyre}
    \begin{align}
    \dot{x}_1 &= -\pi \sin{\pi x_1}\cos{\pi x_2} + u_1 \\
    \dot{x}_2 &= \pi \sin{\pi x_2}\cos{\pi x_1} + u_2
\end{align}
\end{subequations}
$u_1$ and $u_2$ are control inputs to position states $x_1$ and $x_2$. The safe control for \eqref{double_gyre} is obtained by solving the QP-CDF given in \eqref{eq:qp-cdf1}. The center of the obstacle is located at $[1,0]$, and its radius is 0.25 units. The goal is to move from the initial point given by $\bx_0 = [1.5,0.5]$ to the target point at $\bx_T = [0.5,0.5]$ while avoiding the circular obstacle. Fig.~\ref{fig:double_gyre_state_TAC} showcases the safe trajectory of the double-gyre model in the position states.
\begin{figure}[ht]
  \centering
  \includegraphics[width=1\linewidth]{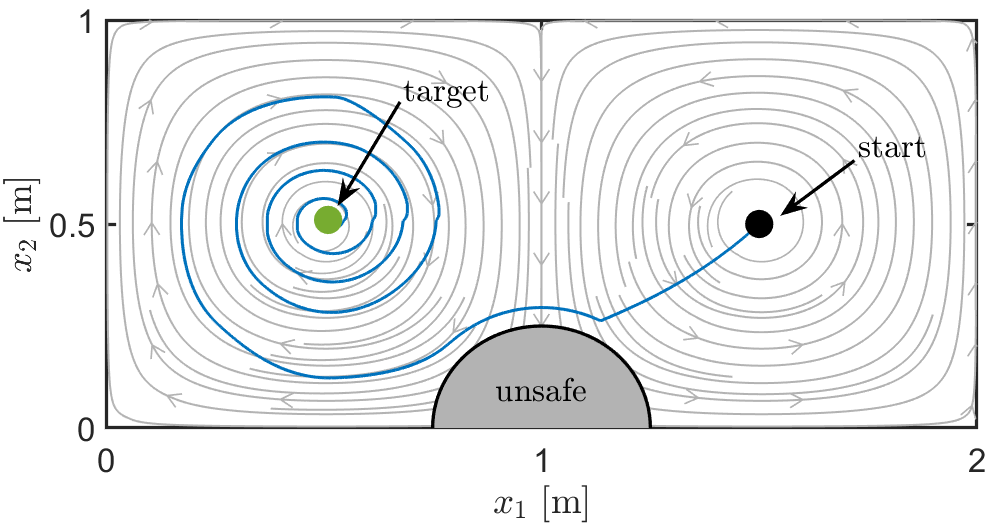}
\caption{\textbf{Double-Gyre flow-field:} Safe navigation trajectory in the position states.}
\label{fig:double_gyre_state_TAC}
\end{figure}
\subsection{Bicycle Model}
Consider the bicycle model as follows:
\begin{subequations}
\begin{align}
    \dot{x}_1 &= v\;\cos\left( \theta + \Phi \right)  \label{eq:bm_x1}\\
    \dot{x}_2 &= v\;\sin\left( \theta + \Phi \right)  \label{eq:bm_x2}\\
    \dot{\theta} &= \frac{v}{L}\;\cos \Phi \tan \Theta  \\
    \dot{\Theta} &= \omega  \\
    \dot{v} &= a \label{eq:bicycle}
\end{align}
\end{subequations}
where 
$$\Phi = \tan^{-1}\left( \frac{l_r\; \tan \Theta}{L}\right)$$
where $x_1,x_2$ represents the position states, $\theta, \Theta, v$ represents the heading angle, steering angle, and linear velocity, respectively. The parameter $l_r$ represents the distance between the rear wheel and the center of mass, whereas the parameter $L$ represents the total length. The control inputs to the bicycle model include steering rate $\omega$, and linear acceleration $a$. Here, we assume the obstacles to avoid are only present in the position states. 
\subsubsection{Uncertainty in the initial conditions}
In this subsection, we consider bounded uncertainty for the position state estimates of the dynamics in the $x_1 - x_2$ space. To design a control for obstacle avoidance in the position states, we first reformulate the \eqref{eq:bm_x1} and \eqref{eq:bm_x2} as follows:
\begin{subequations}
    \begin{align}
        \dot{x}_1 &= u_1 \label{eq:bm_u1} \\
        \dot{x}_2 &= u_2 \label{eq:bm_u2}
    \end{align}
\end{subequations}
We next solve for $u_1$ and $u_2$ using the QP-CDF given in \eqref{finite_robust} to account for the uncertainty in the position state estimate. Here, $p=20$ represents the number of random data samples in the $x_1-x_2$ space selected from the set ${\cal N}_\beta(\bx_0)$ to solve \eqref{finite_robust_1}. The state estimation uncertainty is bounded by $\beta = 0.5$. Here, $c_k(\bx):= \|\bx - o_k\|^2 - r_{1k}^2$ and $b_k(\bx):= \|\bx - o_k\|^2 - r_{2k}^2$ where $o_k$, $r_{1k}$ and $r_{2k}$ are the center, radius and sensing radius of the $k^{th}$ circular obstacle. The control values $u_1$ and $u_2$ so obtained as the solution of the QP-CDF can be used to calculate $\Tilde{v}$ and $\Tilde{\theta} + \Tilde{\Phi}$ as follows:
\begin{align}
    \Tilde{v} = \sqrt{u_1^2 + u_2^2},\;\;\;
    \Tilde{\theta}+ \Tilde{\Phi} = \tan^{-1} \left(\frac{u_2}{u_1} \right)\label{eq:bm_theta_tilde}
\end{align} 
The control $v$ needs to be designed so that $v - \Tilde{v}$ tends to zero.
Therefore, we consider the Lyapunov function given by $0.5(v-\Tilde{v})^2$ which gives us the following control law for $a$:
\begin{align}
    a = \dot{\Tilde{v}} - \sigma_1(v-\Tilde{v}) \label{eq:bi_v}
\end{align}
for some $\sigma_1>0$.

Similarly, the control $\omega$ needs to be designed such that $(\theta + \Phi) - (\Tilde{\theta}+\Tilde{\Phi})$ tends to zero. Therefore, we consider the Lyapunov function given by $1-\cos(\theta + \Phi - \Tilde{\theta}-\Tilde{\Phi})$ which gives us the following control law for $\omega$:
\begin{dmath}
    \omega = \frac{1+\left(\frac{l_r\tan \Theta }{L}\right)^2}{\frac{l_r \sec^2\Theta}{L}}\left[ -\frac{v}{L}\cos\Phi \tan\Theta + \Dot{\Tilde{\theta}}+\Dot{\Tilde{\Phi}}\\-\sigma_2\sin(\theta + \Phi - \Tilde{\theta}-\Tilde{\Phi})\right] \label{eq:bi_omega}
\end{dmath}
for some $\sigma_2>0$.
Fig.~\ref{fig:bicycle_model_position_states} shows the safe navigation trajectory for the bicycle model with some uncertainty in the position state estimates. The radius of both the initial and the target set is 0.5 units, which represents the uncertainty region in the initial position estimates. The obstacle radius is 1.3 units, and the sensing radius is 1.8 units, with the center for the two obstacles located at $[-2,0.5]$ and $[1.7,-0.8]$, respectively. The value of gain $\sigma_1$ and $\sigma_2$ is chosen to be 2 and 30, respectively. The trajectory in Fig.~\ref{fig:bicycle_model_position_states}(a) represents zero uncertainty in the position state estimate. The trajectory in Fig.~\ref{fig:bicycle_model_position_states}(b) illustrates the uncertainty in the position state estimate and the corresponding evolution of the uncertainty state towards the target set. Fig.~\ref{fig:bicycle_model_other_states} represents the state trajectory plot for $\theta$, $\Theta$, and $v$ for both scenarios, i.e., with and without uncertainty in the position state estimates. As observed in Fig.~\ref{fig:bicycle_model_position_states}, The trajectory tries to pass through the gap between the two obstacles to reach the target when there is no uncertainty in the position state estimates. However, in the presence of uncertainty in the estimates, it chooses a different trajectory to account for the safe navigation of all the estimated position states within the uncertainty bound.
\begin{figure}[ht]
  \centering
  \includegraphics[width=1\linewidth]{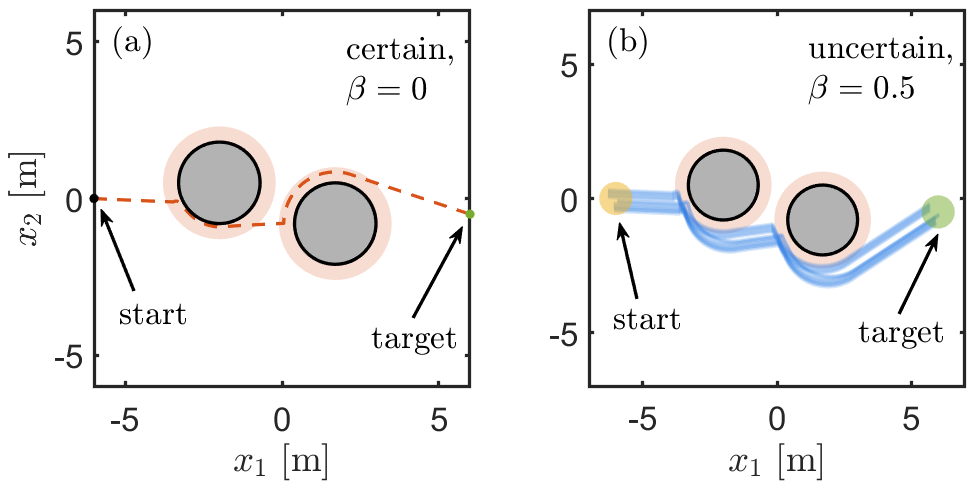}
\caption{\textbf{Bicycle model:} (a) Trajectory plots with zero uncertainty in the state estimate and (b) trajectory plots with uncertainty in the state estimate where uncertainty bound $\beta=0.5$.}
\label{fig:bicycle_model_position_states}
\end{figure}

\begin{figure}[ht]
  \centering
  \includegraphics[width=1\linewidth]{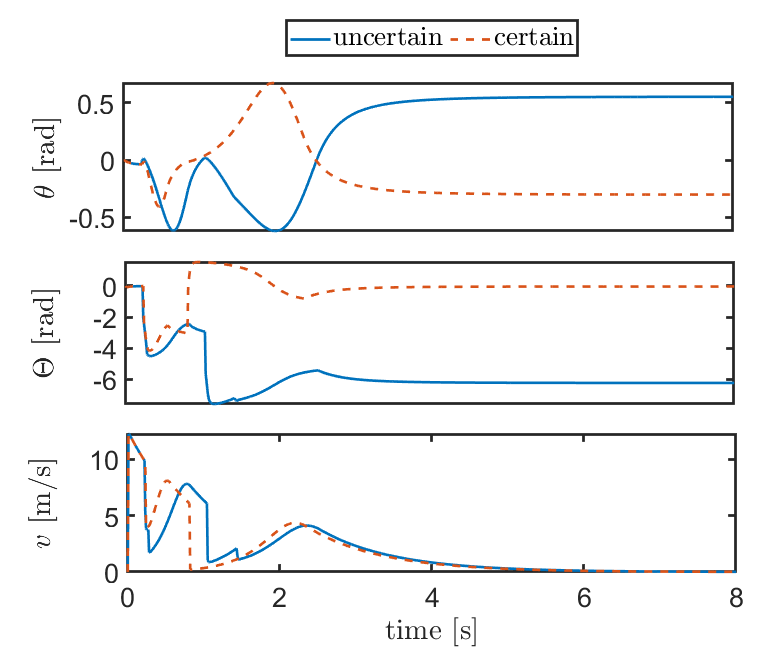}
\caption{\textbf{Bicycle model:} State trajectory plots for $\theta$, $\Theta$, and $v$ for both scenarios i.e., with and without uncertainty in the position state estimates.}
\label{fig:bicycle_model_other_states}
\end{figure}
\subsubsection{Uncertainty in the dynamics}
In this subsection, we consider the bicycle model with uncertainty in the dynamics, as given below.
\begin{subequations}
\begin{align}
    \dot{x}_1 &= v\;\cos\left( \theta + \Phi \right) + \delta_{x_1}  \label{eq:bm_x1_u}\\
    \dot{x}_2 &= v\;\sin\left( \theta + \Phi \right)  + \delta_{x_2} \label{eq:bm_x2_u}\\
    \dot{\theta} &= \frac{v}{L}\;\cos \Phi \tan \Theta + \delta_{\theta}  \\
    \dot{\Theta} &= \omega + \delta_{\Theta}  \\
    \dot{v} &= a + \delta_{v}\label{eq:bicycle_uncertain}
\end{align}
\end{subequations}
where all the uncertainty terms are bounded. The uncertainty could come from either sensor noise or some external disturbance input. To find the safe navigation control law, we substitute $u_1 = v\cos(\theta + \Phi)$ and $u_2 = v\sin(\theta + \Phi)$ to reformulate \eqref{eq:bm_x1_u}-\eqref{eq:bm_x2_u} as given below.
\begin{subequations}\label{eq:bicycle_uncert_dyn_si}
    \begin{align}
        \dot{x}_1 &= u_1 + \delta_{x_1}\label{eq:bicycle_uncert_dyn1_si}\\
        \dot{x}_2 &=u_2 + \delta_{x_2}\label{eq:bicycle_uncert_dyn2_bm}
    \end{align}
\end{subequations}
The safe navigation control law with respect to \eqref{eq:bicycle_uncert_dyn_si} is obtained by solving QP-CDF given in \eqref{eq:robust_qp_cdf} where the value of $\gamma$ is dependent on the bounds of $\delta_{x_1}$ and $\delta_{x_2}$. Next, we calculate $\Tilde{v}$ and $\Tilde{\theta} + \Tilde{\Phi}$ using $u_1$ and $u_2$ as given in \eqref{eq:bm_theta_tilde}. Now, the control $a$ needs to be designed such that $v - \Tilde{v}$ tends to zero.
Therefore, assuming $|\delta_{v}| < \xi_1 $, we consider the Lyapunov function given by $0.5(v-\Tilde{v})^2$ which gives us the following control law for $a$:
\begin{align}
    a = \dot{\Tilde{v}} - \sigma_1(v-\Tilde{v}) -  \xi_1\; \text{sgn}(v-\Tilde{v}) \label{eq:bm_a}
\end{align}
for some $\sigma_1>0$.
Similarly, the control $\omega$ needs to be designed such that $(\theta + \Phi) - (\Tilde{\theta}+\Tilde{\Phi})$ tends to zero.
Therefore, assuming $\max\{|\delta_{\theta}|,|\delta_{\Theta}|\} < \xi_2 $, we consider the Lyapunov function given by $1-\cos(\theta + \Phi - \Tilde{\theta}-\Tilde{\Theta})$ which gives us the following control law for $\omega$:
\begin{dmath}
    \omega = \frac{1+\left(\frac{l_r\tan \Theta }{L}\right)^2}{\frac{l_r \sec^2\Theta}{L}}\left[ -\frac{v}{L}\cos\Phi \tan\Theta + \Dot{\Tilde{\theta}}+\Dot{\Tilde{\Phi}}\\-\sigma_2\sin(\theta + \Phi - \Tilde{\theta}-\Tilde{\Phi}) - \xi_2\; \text{sgn}\left(\sin(\theta + \Phi - \Tilde{\theta}-\Tilde{\Phi})\right)\right] \label{eq:bm_omega}
\end{dmath}
for some $\sigma_2>0$.
Fig.~\ref{fig:bicycle_uncertain_dyn} shows the bicycle model's safe navigation trajectory with some dynamics uncertainty. The obstacle radius is 1.3 units, and the sensing radius is 1.8 units, with the center for the two obstacles located at $[-2,0.5]$ and $[1.5,-1]$, respectively. The value of gain $\sigma_1$ and $\sigma_2$ are chosen to be 2 and 30, respectively. In this example, $|\delta_{x_1}|\le 0.1$ m, $|\delta_{x_1}|\le 0.1$ m, $|\delta_{\theta}|\le 0.5$ radians, $|\delta_{\Theta}|\le 0.5$ radians, and $|\delta_{v}|\le 0.1$ m/s. The trajectory in orange represents zero uncertainty in the dynamics. The blue trajectory represents the bounded uncertainty in all the states of the bicycle model. Fig.~\ref{fig:bicycle_uncertain_dyn_control} displays the state trajectory plot for $\theta$, $\Theta$, and $v$ for both scenarios, i.e., with and without uncertainty in the dynamics.  

\begin{figure}[ht]
  \centering
  \includegraphics[width=.66\linewidth]{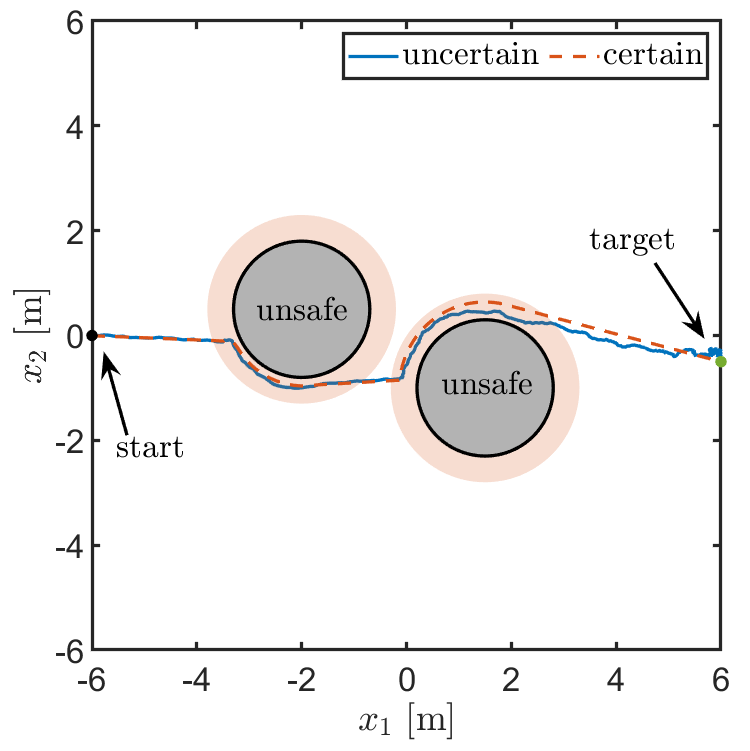}
\caption{\textbf{Bicycle model:} Safe navigation trajectory in the position states in the presence of bounded uncertainty in the dynamics.}
\label{fig:bicycle_uncertain_dyn}
\end{figure}

\begin{figure}[ht]
  \centering
  \includegraphics[width=1\linewidth]{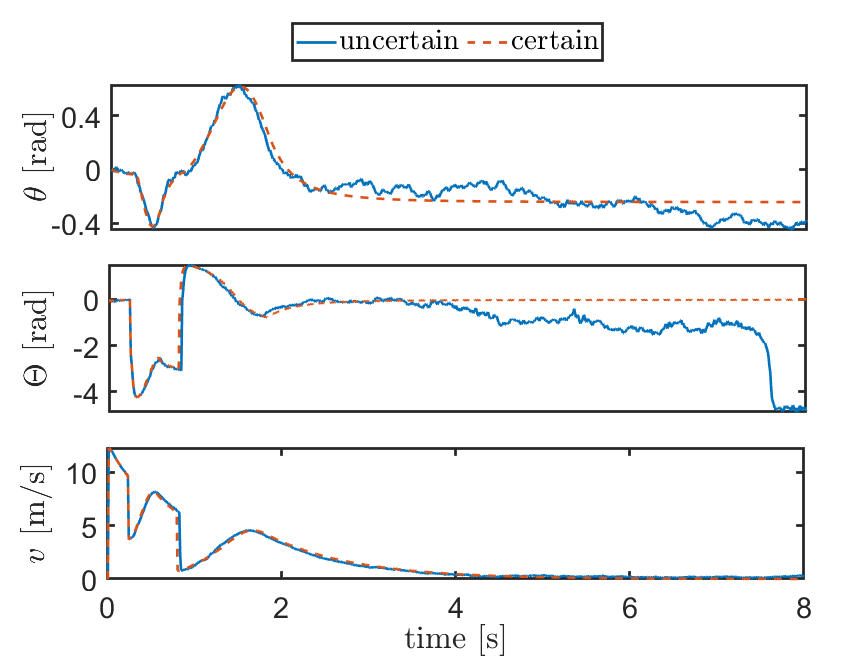}
\caption{\textbf{Bicycle model:} State trajectory plots for $\theta$, $\Theta$, and $v$ for both scenarios i.e., with and without uncertainty in the dynamics.}
\label{fig:bicycle_uncertain_dyn_control}
\end{figure}
\subsection{Autonomous Lane-keeping System}
We consider the lane-keeping (LK) system dynamics with the look-ahead system as shown in Fig.~\ref{fig:lk_model}. We assume that the vehicle has constant longitudinal speed $v_0 = 24$  m/s as we are only interested in the lateral motion of the vehicle. For the LK model, the look-ahead distance makes the vehicle more stable \cite{tan2002automatic}. The dynamics of the LK model have states in terms of lateral position and yaw angle errors with respect to road \cite{rajamani2011vehicle,ames2016control}. Here, the state is $\bx:= [x_1,x_2,x_3,x_4]^\top$ where $x_1$ represents the lateral offset from the lane center at the look-ahead point, $x_2$ is the rate of change of lateral offset from the lane center at the center of gravity (c.g.) of the vehicle, $x_3$ is the heading angle error at the c.g. of the vehicle, and $x_4$ is the yaw rate. We assume $x_3$ will take very small values. The input $u$ to the LK model is the steering angle of the front tires, and $r_d$ represents the desired yaw rate, which is obtained from the curvature of the road by $r_d =  \ \frac {v_0}{R}$ where $R$ is the curvature of the road. The vehicle's mass and moment of inertia are represented by $M$ and $I_z$, respectively. The look-ahead distance is given by $L$, the distance from the c.g. of the vehicle to the front and back tires are denoted by $a$ and $b$, respectively, and $C_r$ and $C_f$ represent tire (stiffness) parameters. The four-state dynamic model for the LK problem is given in \eqref{eq:LK_dynamics}.
\begin{figure}[ht]
  \centering
  \includegraphics[width=0.9\linewidth]{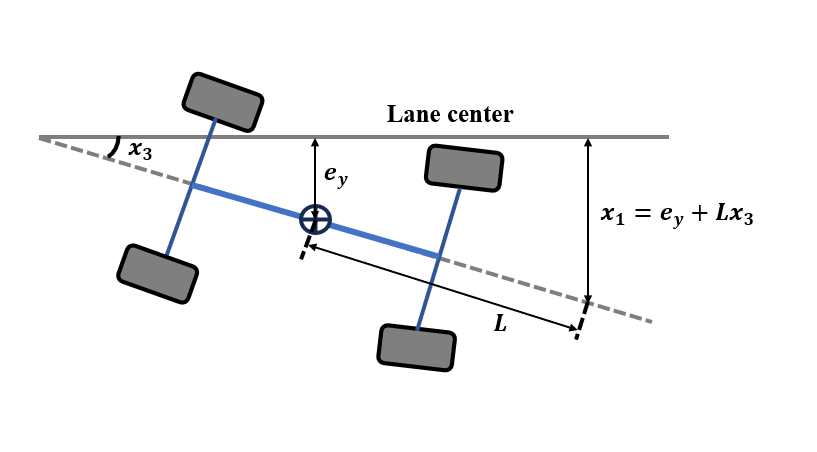}
\caption{Lane keeping system model with look-ahead distance method. }
\label{fig:lk_model}
\end{figure}
\begin{align}
    \Dot{\bx} = \bA\bx + \bB u + \bC r_d \label{eq:LK_dynamics}    
\end{align}
where
\begin{align*}
    \bA &= \begin{bmatrix}
        0&1&0&-L\\
        0&-\frac{2(C_f+C_r)}{Mv_0}&\frac{2(C_f+C_r)}{M}&\frac{2(bC_r - aC_f)}{Mv_0} - 2v_0\\
        0&0&0&-1\\
        0&\frac{2(bC_r - aC_f)}{I_z v_0}&-\frac{2(bC_r - aC_f)}{I_z}&-\frac{2(a^2C_f + b^2C_r)}{I_z v_0}
    \end{bmatrix},\\
    \bB &=\begin{bmatrix}
        0\\ \frac{2C_f}{M}\\ 0 \\ \frac{2aC_f}{I_z}
    \end{bmatrix}, \;\;\; \bC = \begin{bmatrix}
        L \\ v_0 \\ 1 \\ 0
    \end{bmatrix}
\end{align*}
The main goal of the LK problem is to provide an appropriate steering input to keep the vehicle "centered" in the lane. Additionally, the vehicle should satisfy the following hard and acceleration constraints.

\noindent
\textit{Hard Constraint:} The lateral displacement from the center of the lane should be less than the given constant $r_1$:
\begin{align}
    |x_1|\le r_1 \label{eq:LK_hard_const}
\end{align}
\textit{Acceleration Constraint:} The lateral acceleration of the vehicle should be less than a given constant $a_{max}$ due to vehicle limitations.
\begin{align}
    |\Dot{x}_2|\le a_{\textrm{max}} \label{eq:LK_acc_const}
\end{align}
\textit{Encoding LK constraints:} The hard and acceleration constraints are implemented as follows:

\textit{Encoding acceleration constraint:} We know that 
\begin{dmath}
    M\Dot{x}_2= 2C_f\left(u +x_3 - \frac{x_2 + ax_4}{v_0}\right)
    +2C_r\left(x_3 -\frac{x_2 - bx_4}{v_0}\right)\\ - 2Mv_0x_4 + Mv_0r_d\label{eq:LK_acc_const1}
\end{dmath}
Therefore, the acceleration constraint is equivalent to
\begin{align}
    u \in U := \left[ \frac{1}{2C_f}\left( -Ma_{\textrm{max}}+F_0\right),\frac{1}{2C_f}\left(Ma_{\textrm{max}}+F_0\right)\right] \label{eq:LK_acc_u_bounds}
\end{align}
where $F_0 = 2C_f\left(\frac{x_2 + ax_4}{v_0}-x_3\right) + 2C_r\left(\frac{x_2 - bx_4}{v_0}-x_3\right) + 2Mv_0x_4 -Mv_0r_d$.

\textit{Encoding Hard constraint:}
Consider the following:
\begin{align}
    c(\bx) := r_1 - x_1 - 0.5\;\text{sgn}(x_2)\frac{x_2^2}{a_{\textrm{max}}} \label{eq:LK_hx}\\ 
    b(\bx) := r_2 - x_1 - 0.5\;\text{sgn}(x_2)\frac{x_2^2}{a_{\textrm{max}}} \label{eq:LK_sx}
\end{align}
Now, the density function $\rho(\bx)$ for the safe control of the LK problem can be constructed utilizing equations \eqref{eq:bump}-\eqref{eq:rho_S}. Here, $\bx_T = [0,0,0,r_d]^\top$. Also, $r_1 = 0.9, \;r_2 = 0.7$.
\begin{table}[ht]
\centering
\caption{Paramter values used in Lane keeping}
\begin{tabular} {cccc}
\hline
\hline
\\
$M$ & 1589 kg & $r_2$ & 0.7 $\textrm{m}$\\
$a$ & 1.57 m & $C_f$ & 90000 N/rad\\
$b$ & 1.05 m & $C_r$ & 60000 N/rad\\
$L$ & 20 m & $I_z$ & 1765 $\textrm{m}^2$ $\cdot$ kg\\
$r_1$ & 0.9 m & $a_{\textrm{max}}$ & $0.3 \times 9.8 \;\textrm{m}/\textrm{s}^2$\\ 
\\
\hline
\hline
\end{tabular}
\label{table:lane_keeping}
\end{table}

Next, we consider an LK model with some dynamic uncertainty. The uncertainty appears in the dynamics in the form of external disturbance input. We consider two external disturbances (lateral force and moment). Lateral force disturbance $d_F$ represents wind gusts and bank angle. Moment disturbance $d_M$ usually arises because of hydroplaning and wind gust situations. The LK model dynamics in the presence of external disturbances is given by \eqref{eq:LK_dynamics_disturb}.     

\begin{align}
    \Dot{\bx} = \bA\bx + \bB u + \bC r_d + \bD_1 d_F + \bD_2 d_M \label{eq:LK_dynamics_disturb}  
\end{align}
where,
\begin{align*}
    \bD_1 = \begin{bmatrix}
        0\\ \frac{1}{M}\\ 0\\ 0
    \end{bmatrix},\;\;
    \bD_2 = \begin{bmatrix}
        0\\ 0\\ 0\\ \frac{1}{I_z}
    \end{bmatrix}
\end{align*}
We consider the scenario where the LK model experiences both these external disturbances. The bounds for $d_{F}$ and $d_M$ equal $\pm 800$ N and $\pm 400$ Nm, respectively. Fig. \ref{fig:lane_keeping_robust} shows the simulation plots for robust LK example. Fig.~\ref{fig:lane_keeping_robust}(a) shows lateral position plots for ten scenarios, each experiencing a different constant external disturbance in the form of lateral force and moment. Fig.~\ref{fig:lane_keeping_robust}(b) shows the corresponding lateral acceleration plots in units of $g=9.8 \;m/s^2$. The magnitude of the external disturbance for each scenario lies within the bounds. In Fig.~\ref{fig:lane_keeping_robust}(a), we observe that the lateral position plots converge to zero in each scenario to maintain their position in the middle of the lane. Here, $r_1$ and $r_2$ represent the lane edges and sensing region boundaries, respectively. Similarly, in Fig.~\ref{fig:lane_keeping_robust}(b), the acceleration plots lie well within the limits of the vehicle dynamics.

\begin{figure}[ht]
  \centering
  \includegraphics[width=1\linewidth]{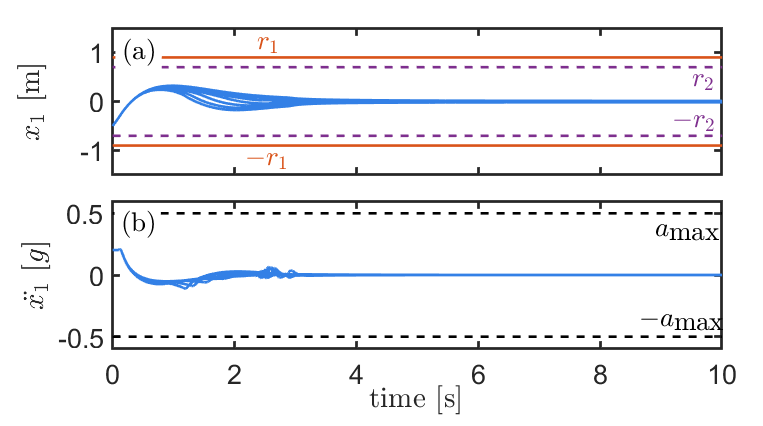}
\caption{\textbf{Lane-keeping:} (a) shows lateral position plots for 10 scenarios each experiencing a different constant external disturbance in the form of lateral force and moment. (b) shows the corresponding lateral acceleration plots in units of $g=9.8 \;m/s^2$.}
\label{fig:lane_keeping_robust}
\end{figure}

\section{Conclusion}\label{section_conclusion}
A novel approach for safe control design based on density function is presented in this paper. We introduce the control density function (CDF) to synthesize a safe controller for nonlinear systems. The occupancy-based physical interpretation of the density function allows us to provide an analytical construction of the CDF, which is used as a constraint in the quadratic program (QP) for the safe control design. We also consider uncertainty in the system dynamics and initial conditions for robust safe navigation. Finally, we have validated our framework by providing simulation results for the bicycle model, double-gyre fluid flow field, and lane-keeping example.

\section*{Appendix}

The proof of Theorem \ref{thm:11} relies on the following Lemma which was first presented in \cite{moyalan2022navigation}.
\begin{lemma}\label{lemma_1}
    If
    \begin{align}
        \int_0^{\infty}\int_{\bX_0}\mathds{1}_{\bX_u}(s_t(\bx)) d\bx \;dt=0 \label{eq:lemma_1}
    \end{align}
    then
    \begin{align}
        \int_{\bX_0}\mathds{1}_{\bX_u}(s_t(\bx)) d\bx = 0 \;\;\;\forall t\ge0 \label{eq:lemma_2}
    \end{align}
    i.e., the amount of time system trajectories spend in set $\bX_u$ starting from the positive measure set $\bX_0$ is equal to zero.
\end{lemma}

\begin{proof}
The proof is done by the method of contradiction. Let us assume \eqref{eq:lemma_2} is not true. Then there exists a time $\bar{t} \ge 0$ such that 
$$ \int_{\bX_0}\mathds{1}_{\bX_u}(s_{\bar t}(\bx)) d\bx   >0.$$
Now, from the continuity of the solution of the differential equation, there exists a $\Delta$ such that
$$ \int_{\bar t}^{\bar t+\Delta}\int_{\bX_0}\mathds{1}_{\bX_u}(\bs_t(\bx)) d\bx \;dt >0$$
Therefore,
\begin{align*}
    0 <&  \int_{\bar t}^{\bar t+\Delta}\int_{\bX_0}\mathds{1}_{\bX_u}(\bs_t(\bx)) d\bx \;dt \\
    \leq& \int_0^{\infty}\int_{\bX_0}\mathds{1}_{\bX_u}(s_t(\bx)) d\bx \;dt=0.
\end{align*}
Hence, we arrive at a contradiction.
\end{proof}

\begin{proof}[{\bf Proof of Theorem \ref{thm:11}}] We begin with the proof of a.e. convergence and later showcase the proof of obstacle avoidance. The steps for the proof of convergence follow similar lines to the proof of Theorem 1 in \cite{zheng2023safe}.
Let us consider the following:
\begin{align}
    \nabla \cdot ((\bff (\bx) + \bg(\bx)\bk)\rho) = h_0(\bx) \label{eq:proof_1}
\end{align}
where $h_0(\bx) \ge 0$ and $h_0(\bx)\ge \lambda >0$ for $\bx \in \bX_0$. Now, through the method of characteristics, the solution $\rho(\bx)$ for the linear PDE (\ref{eq:proof_1}) can be written as follows \cite{rajaram2010stability}:
\begin{align}
    \rho(\bx) = \int^{\infty}_0 h_0(s_{-t}(\bx))\left | \frac{\partial s_{-t}(\bx)}{\partial \bx}\right|dt,\label{eq:proof_2}
\end{align}
where $|\cdot|$ represents the determinant. This can be easily verified by simple substitution of \eqref{eq:proof_2} in \eqref{eq:proof_1} and using the fact that
\begin{align}
    \lim_{t\rightarrow \infty}h_0(s_{-t}(\bx))\left | \frac{\partial s_{-t}(\bx)}{\partial \bx}\right| = 0 .\label{eq:proof_3}
\end{align}
The limit in \eqref{eq:proof_3} is the consequence of $\rho(\bx)$ being bounded in $\Bar{\bX}$ and using Barbalat's Lemma which states that if $f(t) \in \mathcal{C}^1, \lim_{t\rightarrow \infty}f(t)=c$ for some constant $c>0$ and, if $f'(t)$ is uniformly continuous, then $\lim_{t \rightarrow \infty}f'(t)=0$. The term inside the integral in \eqref{eq:proof_2} can be written using the linear Perron-Frobenius (P-F) operator \cite{lasota1998chaos} defined as follows:
\begin{align}
    [\mathbb{P}_t h_0](\bx):= h_0(s_{-t}(\bx))\left | \frac{\partial s_{-t}(\bx)}{\partial \bx}\right|.\label{eq:proof_4}
\end{align}
Therefore (\ref{eq:proof_2}) can be written compactly using the above definition of the P-F operator as
\begin{align}
    \rho(\bx) = \int^{\infty}_0  [\mathbb{P}_t h_0](\bx)dt .\label{eq:proof_5}
\end{align}
Now, utilizing \eqref{eq:proof_3}, we can write
\begin{align}
    \lim_{t \rightarrow \infty}[\mathbb{P}_t h_0](\bx) = 0 \implies \lim_{t \rightarrow \infty}[\mathbb{P}_t \mathds{1}_{\bX_0}](\bx) = 0 \label{eq:proof_6}
\end{align}
where $\mathds{1}_{\bX_0}$ represents the indicator function for $\bX_0$. Now, (\ref{eq:proof_6}) follows because $h_0(\bx) \ge \lambda > 0\;\;\forall \bx\in\bX_0$ and using dominated convergence theorem. Here, $\mathds{1}_{\bX_0}$ represents the indicator function for $\bX_0$. Now, for any $A \subseteq \Bar{\bX}$, we have
\begin{align}
    \int_A [\mathbb{P}_t \mathds{1}_{\bX_0}](\bx)d\bx =& \int_{\Bar{\bX}}[\mathbb{P}_t \mathds{1}_{\bX_0}](\bx)\mathds{1}_{A}(\bx)d\bx\nonumber\\
    =&\int_{\Bar{\bX}}\mathds{1}_{\bX_0}(\bx)\mathds{1}_{A}(s_{t}(\bx))d\bx\label{eq:proof_7}
\end{align}
This can be observed by using the definition of P-F operator in \eqref{eq:proof_4} and doing the change of variables in integration such as $\by = s_{-t}(\bx)$ and $d\by=\left | \frac{\partial s_{-t}(\bx)}{\partial \bx}\right|d\bx$ and relabelling. The right hand side of \eqref{eq:proof_7} can be seen as follows:
\begin{align*}
    \int_A [\mathbb{P}_t \mathds{1}_{\bX_0}](\bx)d\bx = m\{\bx \in \bX_0 : s_t(\bx)\in A\}.
\end{align*}
Therefore, using \eqref{eq:proof_6}, we observe that
\begin{align*}
    0 =  \int_A [\lim_{t \rightarrow \infty}\mathbb{P}_t \mathds{1}_{\bX_0}](\bx)d\bx = m\{\bx \in \bX_0 : \lim_{t \rightarrow \infty}s_t(\bx)\in A\}.
\end{align*}
The above statement can be generalized for any measurable Lebesgue set $A \subseteq \Bar{\bX}$. Therefore,
\begin{align*}
    m\{\bx \in \bX_0 : \lim_{t \rightarrow \infty}s_t(\bx)\neq 0\} = 0.
\end{align*}
Next, we will show the proof of obstacle avoidance. Now, from the construction of the density function, we know that $\rho(\bx) = 0 \;\forall \bx \in \bX_u$. Therefore, we conclude the following,
\begin{align}
    \int_{\bX_u}\int^{\infty}_0[\mathbb{P}_t \mathds{1}_{\bX_0}](\bx)dt d\bx \le \int_{\bX_u}\rho(\bx)d\bx = 0. \label{eq:proof_8}
\end{align}
Utilizing the Markov property of the P-F operator and the fact that indicator functions are non-negative functions, we can rewrite \eqref{eq:proof_8} as follows:
\begin{align}    \int_{\bX_u}\int^{\infty}_0[\mathbb{P}_t \mathds{1}_{\bX_0}](\bx)dt d\bx  = 0. \label{eq:proof_9}
\end{align}
Now, doing the change of variables in integration such as $\by = s_{-t}(\bx)$ and $d\by=\left | \frac{\partial s_{-t}(\bx)}{\partial \bx}\right|d\bx$ and relabelling, the left-hand side of the \eqref{eq:proof_9} can be written as follows:
\begin{align}
\int_{\bX_u}\int^{\infty}_0[\mathbb{P}_t \mathds{1}_{\bX_0}]&(\bx)dt d\bx = \int_{\Bar{\bX}}\int^{\infty}_0 [\mathbb{P}_t \mathds{1}_{\bX_0}](\bx)\mathds{1}_{\bX_u}(\bx)dt d\bx \nonumber\\
=& \int_0^{\infty}\int_{\Bar{\bX}}\mathds{1}_{\bX_0}(\bx)\mathds{1}_{\bX_u}(s_t(\bx)) d\bx \;dt\nonumber\\
=& \int_0^{\infty}\int_{\bX_0}\mathds{1}_{\bX_u}(s_t(\bx)) d\bx \;dt=0.\label{eq:proof_10}
\end{align}
Now, from \eqref{eq:proof_10} and using Lemma \ref{lemma_1}, we can conclude the following:
\begin{align*}
    \int_{\bX_0}\mathds{1}_{\bX_u}(s_t(\bx)) d\bx = 0 \;\;\;\forall t\ge0   
\end{align*}
\end{proof}

\begin{proof}[{\bf Proof of Theorem \ref{theorem_statefeedback}}]
    The constraint in \eqref{constraint_1} can be rewritten as follows:

\begin{align}
\nabla \cdot (\bff(\bx)\rho) + \sum_{i=1}^m \nabla \cdot (g_i(\bx)\rho\; u_i) \ge 0 \label{eq:main_constraint}
\end{align}

Next, let us consider $\rho u_i := \bar{\rho}_i$. Now, we know that the infinitesimal generator for the P-F operator with respect to $\bff(\bx)$ is given as follows \cite{moyalan2023data}: 
\begin{align}
\lim_{t\to 0}\frac{\mathbb{P}_t\rho-\rho}{t}=-\nabla \cdot (\bff(\bx) \rho(\bx))=: {\cal P}_{\bff}\rho. \label{PF_generator}
\end{align} 
Therefore the definition of P-F generator in \eqref{eq:main_constraint}, we can rewrite the QP-CDF in \eqref{QP-CDF_base} as:
\begin{subequations}\label{QP-CDF_v2}
\begin{align}
    \bu^{\star}(\bx) = \argmin_{\bar{\brho}} \;\;\;\;\bar{\brho}^\top \bar{\bH}(\bx) \bar{\brho}+ \bar{\bJ}(\bx)\bar{\brho}\\
    \text{s.t.}\;\;\;\;\;-\mathcal{P}_{\bff} \rho - \sum_{i=1}^m \mathcal{P}_{g_i} \bar{\rho}_i \ge 0 \label{eq:main_constraint_v1}
\end{align}
\end{subequations}
where $\bar{\brho} = [\bar{\rho}_1,\dots,\bar{\rho}_m]^\top$,$\bar{\bH}(\bx) = \bH(\bx)\bH_1(\bx)$ and $\bar{\bJ}(\bx) = \bJ_1(\bx)\bJ(\bx)$, and
\begin{align*}
    \bH_1(\bx) =\begin{bmatrix}
    \frac{1}{\rho^2}&0&0\\
    0 &\ddots &0 \\
    0&0 &\frac{1}{\rho^2}
\end{bmatrix}, \;\;\;\;\;\;\;
    \bJ_1(\bx) =\begin{bmatrix}
    \frac{1}{\rho}&0&0\\
    0 &\ddots &0 \\
    0&0 &\frac{1}{\rho}
\end{bmatrix}
\end{align*}
Let $\by := [\mathcal{P}_{g_1},\dots,\mathcal{P}_{g_m}]^\top$ and $q(\bx):=-\mathcal{P}_{\bff} \rho(\bx)$. Therefore, \eqref{eq:main_constraint_v1} can be rewritten as follows:
\begin{align}
    \by^\top \bar{\brho} \le q(\bx) \label{transform1}
\end{align}
Let $\bar{\by} = \bar{\bH}^{-1}(\bx)\by$ and $\bar{q}(\bx) = q(\bx) - \by^\top \hat{\bu}$ where $\hat{\bu} = -\bar{\bH}^{-1}(\bx)\bar{\bJ}(\bx)$. Also, let $\bv = \bar{\brho}-\hat{\bu}$ and $<\bv,\bv>:=\bv^\top \bar{\bH}(\bx)\bv$. Therefore, we can rewrite \eqref{QP-CDF_v2} as follows:
\begin{subequations} \label{eq:v_star}
    \begin{align}
    \bv^{\star}(\bx) = \argmin_{ \bv \in \mathds{R}^{m+1}} \;\;\;\;<\bv,\bv>\\
    \text{s.t.}\;\;\;\;\;<\bar{\by}(\bx),\bv> \le \bar{q}(\bx) \label{eq:main_constraint_v2}
\end{align}
\end{subequations}
with $\bar{\brho}^{\star}(\bx) = \bv^{\star}(\bx)+\hat{\bu}(\bx)$.

Therefore, we can rewrite $\bv^{\star}(\bx)$ as follows:
\begin{align}
    \bv^{\star}(\bx) = \lambda(\bx)\bar{\by}(\bx)
\end{align}
where $\lambda(\bx)$ is obtained as the solution to
\begin{align}
    <\bar{\by}(\bx),\bar{\by}(\bx)>\lambda(\bx) \le \bar{q}(\bx), \nonumber \\
    \lambda(\bx) \le 0, \nonumber \\
    <\bar{\by}(\bx),\bar{\by}(\bx)>\lambda(\bx) < \bar{q}(\bx) \implies \lambda(\bx) =0.
\end{align}
Next, we define the Lipschtiz continuous function

\begin{align} \label{omega}
    \omega(r) = \begin{cases}
    0, &\text{if}\; r>0,\\
    r, &\text{if}\; r\le 0, \;\;r \in \mathds{R}
    \end{cases}
\end{align}
Therefore, 
\begin{align}
    \lambda(\bx) = \frac{\omega(\bar{q}(\bx))}{<\bar{\by}(\bx),\bar{\by}(\bx)>} \label{vstar_solution}
\end{align}
Therefore, 
\begin{align}
    \bar{\brho}^{\star}(\bx) = \frac{\omega(\bar{q}(\bx))}{<\bar{\by}(\bx),\bar{\by}(\bx)>} - \bar{\bH}^{-1}(\bx)\bar{\bJ}(\bx) \label{vstar_solution1}
\end{align}
and $\bu^{\star}(\bx) = \frac{\bar{\brho}^{\star}(\bx)}{\rho(\bx)}$.

Based on the properties for the composition and the product of $\mathcal{C}^1$ continuous functions, we observe that the RHS of \eqref{vstar_solution1} is $\mathcal{C}^1$ except when $\bar{q}(\bx) = 0$, which implies that $\bar{\brho}^{\star}(\bx)$ and subsequently $\bu^{\star}(\bx)$ obtained from the QP-CDF are $\mathcal{C}^1$ continuous functions for almost all $\bx \in \bar{\bX}$.
\end{proof}

\begin{proof}[{\bf Proof of Lemma \ref{thm:1}}]
    Let $\bX_0 \subset \bar\bX$ be the set of initial conditions.  Let $Z \subset \bX_0$. Therefore, from the definition of the density function
    \begin{dmath}        \int_{s_t(Z)}\rho(\bx)d\bx - \int_{Z}\rho(\bx)d\bx =\nonumber\\
        \int_0^t \int_{s_{\tau}(Z)}\left[\nabla\cdot(\bF(\bx,\tau)\rho) \right]d\bx d\tau \label{eq:proof_thm1}
    \end{dmath}
First, we will prove that the system trajectories obtained from \eqref{time_varying} avoid $\bX_u$. This proof is done through the method of contradiction. Let there exists an initial condition $\bx_0 \in \bX_0$ such that $s_T(\bx_0) \in \bX_u$ for some $T>0$ and $s_t(\bx_0) \in \bar\bX$ for $t \in [0,T]$. Let $Z \subset \bX_0$ be a small neighbourhood of $\bx_0$ such that $s_T(Z) \subset \bX_u$ for some $T>0$ and $s_t(Z) \subset \bar\bX$ for $t \in [0,T]$. Therefore, from the construction of $\rho$, at $t=T$, the first term on LHS of \eqref{eq:proof_thm1} is zero. Also, $Z \subset \bX_0$. Therefore, the second term on the LHS of \eqref{eq:proof_thm1} is positive. Therefore, overall the LHS of \eqref{eq:proof_thm1} is negative. The RHS of \eqref{eq:proof_thm1} from the constraint of the QP-CDF is positive. This gives a contradiction.

        The part of the proof for the convergence of system trajectories to the target set $\bX_T$ follows using the result of Theorem 2.2 from \cite{masubuchi2021lyapunov}.
\end{proof}

\section*{references}
\bibliographystyle{IEEEtran}
\bibliography{reference,Umesh_ref}
\begin{IEEEbiography}[{\includegraphics[width=1in,height=1.25in,clip,keepaspectratio]{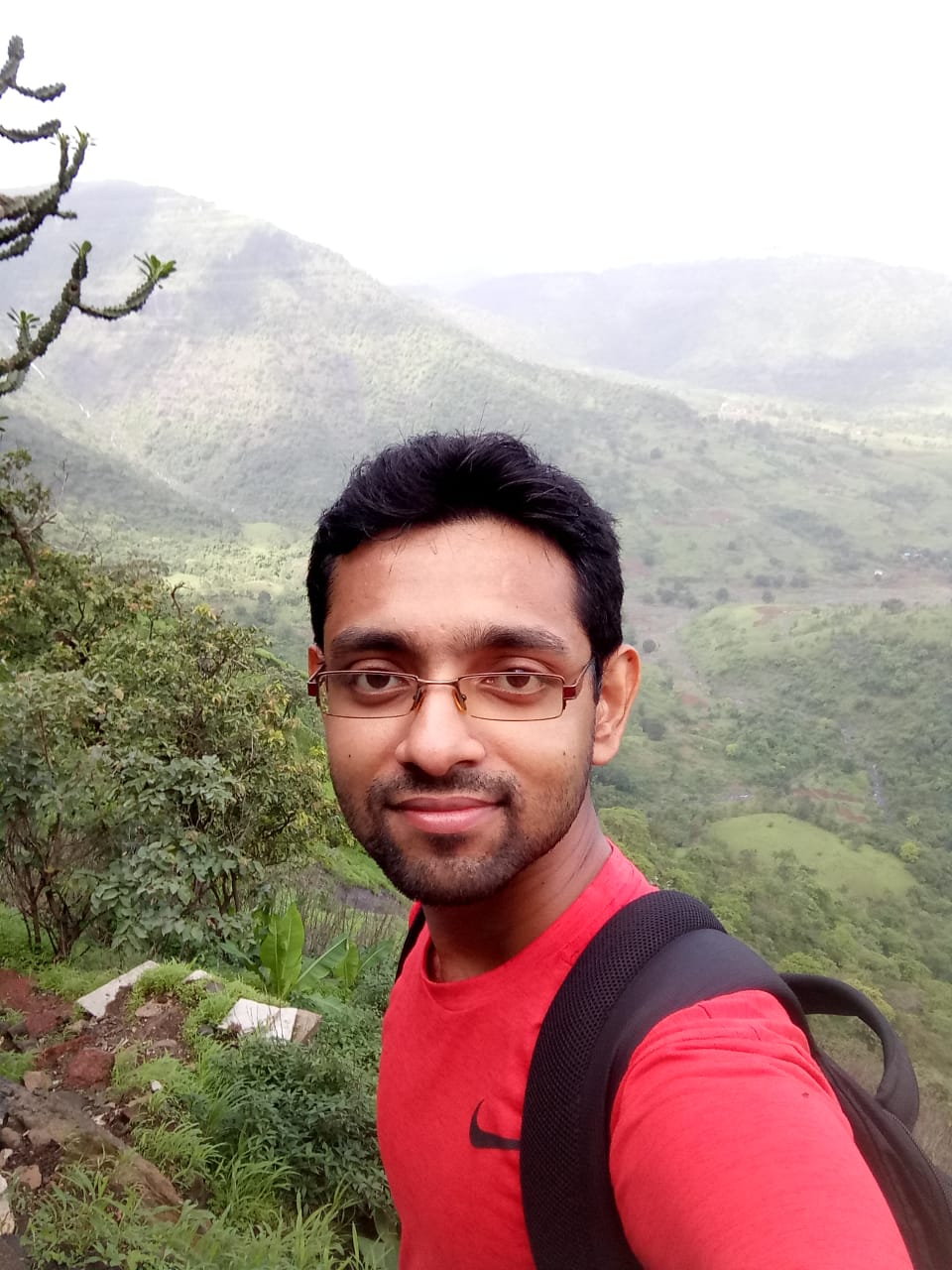}}]{Joseph Moyalan} received his B.E. degree in electronics and telecommunication engineering from Mumbai
University, India, in 2016, and his M. Tech. degree in
electrical engineering (specialization in control systems)
from Veermata Jijabai Technological Institute, Mumbai,
India, in 2019. He is currently pursuing his Ph.D. degree
in the Mechanical Engineering Department at Clemson University, Clemson, SC, USA. His current research interests
include nonlinear optimal control, convex optimization, linear operators, and data-driven control with
application to vehicle autonomy.
\end{IEEEbiography}
\begin{IEEEbiography}
[{\includegraphics[width=1in,height=1.25in,clip,keepaspectratio]{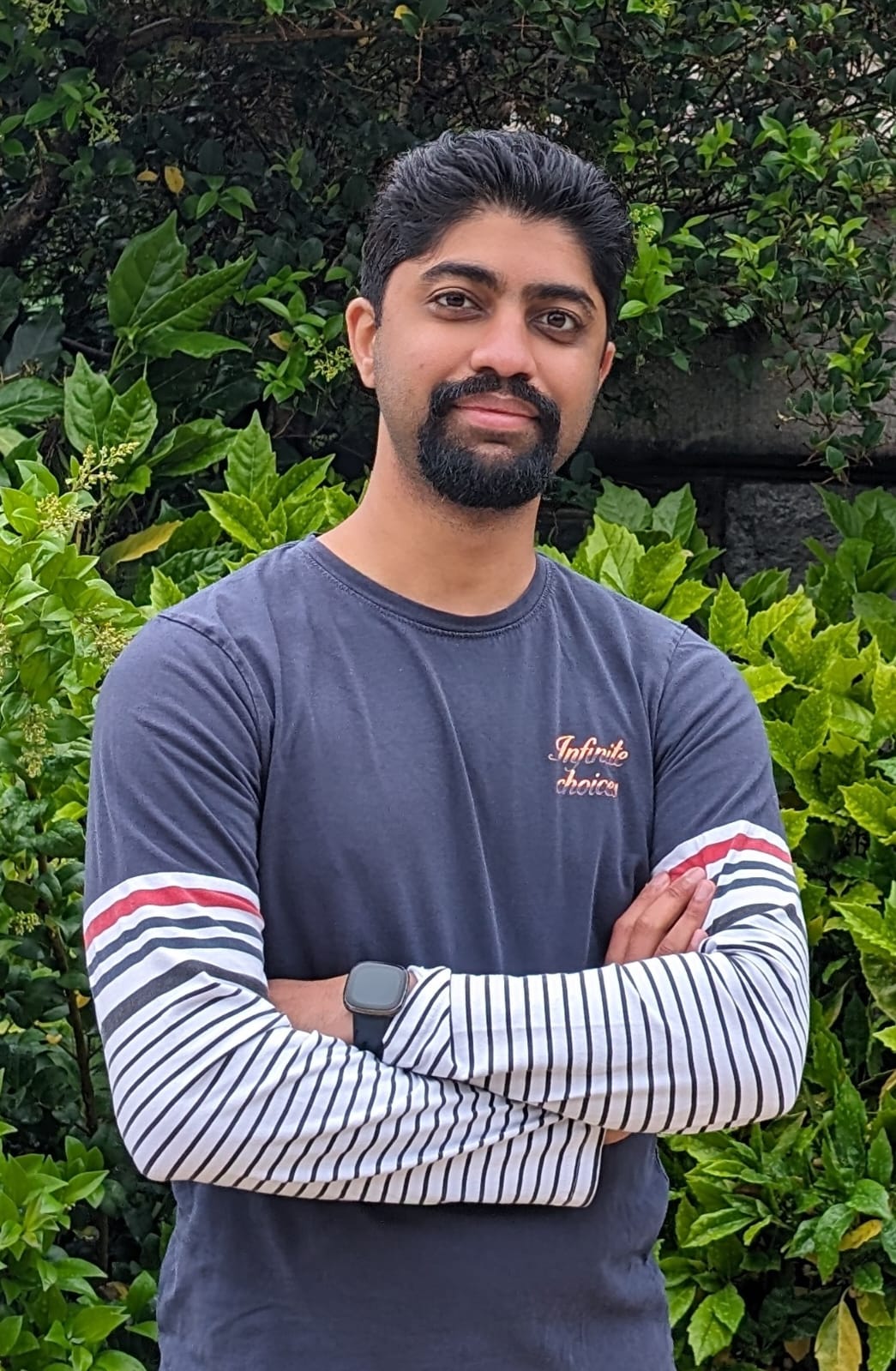}}]{Sriram S. K. S Narayanan} received his B.Tech degree in mechanical engineering from SASTRA University, India, in 2017, and his MS. degree in mechanical engineering from the University of South Florida, Tampa, FL, USA in 2021. He is currently pursuing his Ph.D. degree in the Mechanical Engineering Department at Clemson University, Clemson, SC, USA. His current research interests include motion planning, optimal control, and model predictive control with applications to legged robots and autonomous vehicles.
\end{IEEEbiography}
\begin{IEEEbiography}
[{\includegraphics[width=1in,height=1.25in,clip,keepaspectratio]{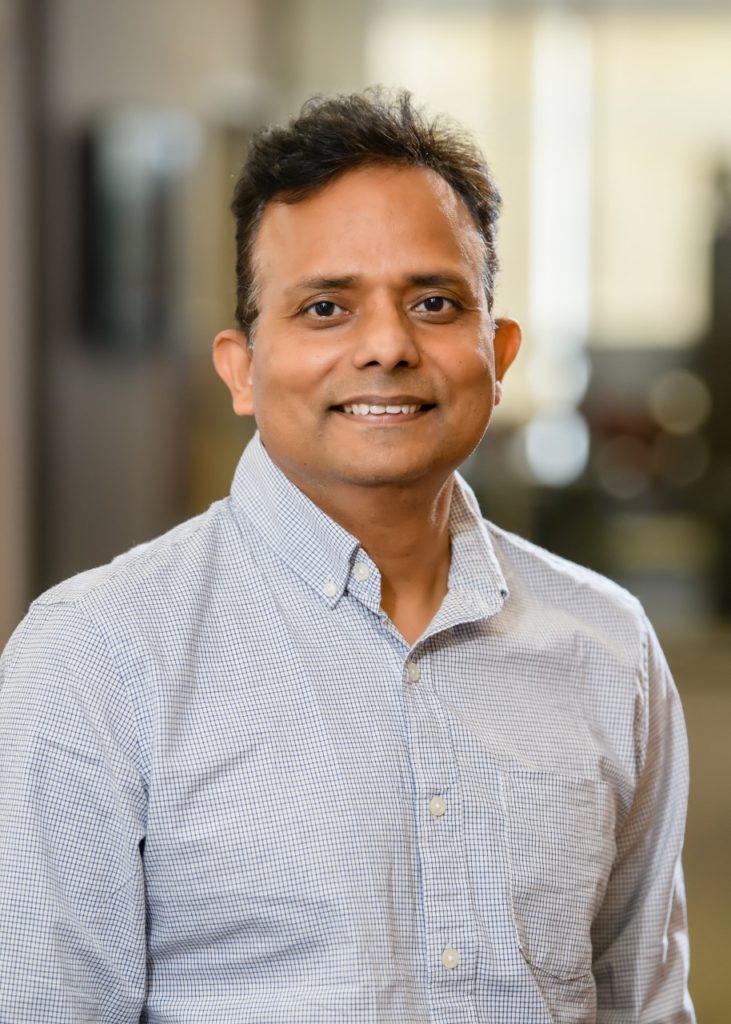}}]{Umesh Vaidya}(M’07, SM'19)  received the Ph.D. degree in mechanical engineering from the University of California at Santa Barbara, Santa Barbara, CA, in
2004. He was a Research Engineer at the United Technologies Research Center (UTRC), East Hartford, CT, USA. He is currently a professor in the Department of Mechanical Engineering, at Clemson University, S.C., USA. Before joining Clemson University in 2019, and since 2006, he was a faculty with the Department of Electrical and Computer Engineering at Iowa State University. He is the recipient of the 2012 National Science Foundation CAREER award. His current research interests include dynamical systems and control theory with applications to power grids and robotics.
\end{IEEEbiography}

\end{document}